\documentclass[11pt,a4paper]{article}
\usepackage{silence} 
\usepackage{jheppub}
\usepackage{graphicx}
\usepackage{orcidlink} 
\usepackage[makeroom]{cancel}
\usepackage{enumitem}
\usepackage{amssymb,amsmath,mathrsfs,bm,bbm,mathtools} 
\usepackage{amsthm}
\definecolor{ao(english)}{rgb}{0.0, 0.5, 0.0}
\hypersetup{colorlinks=true,linkcolor=blue,urlcolor=blue,citecolor=ao(english)}
\usepackage{multirow}

\newcommand\nn{\nonumber\\}
\newcommand\non{\nonumber}
\newcommand{\bma}{\left(\begin{array}}
\newcommand{\ema}{\end{array}\right)}
\newcommand{\be}{\begin{equation}}
\newcommand{\ee}{\end{equation}}
\newcommand{\ben}{\begin{equation*}}
\newcommand{\een}{\end{equation*}}
\newcommand{\ba}{\begin{eqnarray}}
\newcommand{\ea}{\end{eqnarray}}
\newcommand{\ban}{\begin{eqnarray*}}
\newcommand{\ean}{\end{eqnarray*}}
\newcommand{\bs}{\begin{subequations}}
\newcommand{\es}{\end{subequations}}
\newcommand{\bc}{\begin{center}}
\newcommand{\ec}{\end{center}}
\newcommand{\ve}{\varepsilon}

\newcommand{\Pl}{{\text{\tiny Pl}}}
\newcommand{\lp}{\ell_\Pl}
\newcommand{\tpl}{t_\Pl}

\newcommand{\Mpl}{M_\Pl}

\newcommand{\au}[2]{#1.~#2}
\newcommand{\arX}[1]{\href{http://arxiv.org/abs/#1}{{\cob arXiv:#1}}}
\newcommand{\oarX}[1]{\href{http://arxiv.org/abs/#1}{{\cob #1}}}

\newcommand{\book}[5]{\emph{#1}, #2, #3, #4 (#5)}
\newcommand{\books}[4]{\emph{#1}, #2, #3 (#4)} 
\newcommand{\doin}[6]{\href{http://dx.doi.org/#1}{{\cob {\it #2 #3} {\bf #4} (#6) #5}}}
\newcommand{\doinn}[5]{\href{http://dx.doi.org/#1}{{\cob {\it #2} {\bf #3} (#5) #4}}}
\newcommand{\doij}[5]{\href{http://dx.doi.org/#1}{{\cob {\it #2} {\bf #3} (#5) #4}}}

\newcommand{\ndoinn}[5]{\href{#1}{{\cob {\it #2} {\bf #3} (#5) #4}}}
\newcommand{\procsinm}[5]{in \emph{#1}, #2 (eds.), #3, #4 (#5)}

\newcommand{\procm}[6]{in \emph{#1}, #2 (eds.), #3, #4, #5 (#6)}
\newcommand{\tia}[1]{\textit{#1},}

\newcommand{\boxd}[1]{\boxed{\phantom{\Biggl(}#1\phantom{\Biggl)}}}

\renewcommand{\leq}{\leqslant}
\renewcommand{\geq}{\geqslant}
\newcommand{\Eq}[1]{(\ref{#1})}
\newcommand{\Eqq}[1]{eq.~(\ref{#1})}
\newcommand{\Eqqs}[1]{eqs.~(\ref{#1})}
\def\rme{e}
\def\rmd{d}
\def\rmi{i}

\def\Re{\text{Re}}

\def\erf{{\rm erf}}
\def\a{\alpha}
\def\b{\beta}
\def\de{\delta}

\def\g{\gamma}
\def\la{\lambda}

\def\ve{\varepsilon}
\def\Om{\Omega}
\def\om{\omega}
\def\G{\Gamma}
\def\t{\tau}
\def\s{\sigma}

\def\vp{\varphi}

\def\B{\Box}
\def\H{{\rm H}}
\def\mst{M_*}
\def\Lst{\Lambda_*}
\def\Lz{\Lambda_*}
\def\Est{\Lambda_{\rm hd}}
\def\lst{\ell_*}

\def\Hes{{\bf\Delta}}

\def\cA{\mathcal{A}}

\def\cD{\mathcal{D}}

\def\cL{\mathcal{L}}

\def\cR{\mathcal{R}}

\def\p{\partial}

\def\cob{\color{blue}}

\newtheorem{theo}{Theorem}

\begin{document}
	
\renewcommand{\thefootnote}{\fnsymbol{footnote}}

\title{Early universe in quantum gravity} 

\author[a,b]{Leonardo Modesto\,\orcidlink{0000-0003-2783-8797}}
\emailAdd{leonardo.modesto@unica.it}
\affiliation[a]{Dipartimento di Fisica, Universit\`a di Cagliari, Cittadella Universitaria, 09042 Monserrato, Italy}
\affiliation[b]{Department of Physics, Southern University of Science and Technology, Shenzhen 518055, China}

\author[c,*]{and Gianluca Calcagni\,\orcidlink{0000-0003-2631-4588}\note{Corresponding author.}}
\emailAdd{g.calcagni@csic.es}
\affiliation[c]{Instituto de Estructura de la Materia, CSIC, Serrano 121, 28006 Madrid, Spain}


\abstract{We present a new picture of the early universe in finite nonlocal quantum gravity, which is Weyl invariant at the classical and quantum levels. The high-energy regime of the theory consists of two phases, a Weyl invariant trans-Planckian phase and a post-Planckian or Higgs phase described by an action quadratic in the Ricci tensor and where the cosmos evolves according to the standard radiation-dominated model. In the first phase, all the issues of the hot big bang such as the singularity, flatness, and horizon problems find a universal and simple non-inflationary solution by means of Weyl invariance, regardless of the microscopic details of the theory. In the second phase, once Weyl symmetry is spontaneously broken, primordial perturbations are generated around a background that asymptotically evolves as a radiation-dominated flat Friedmann--Lema\^{i}tre--Robertson--Walker universe. 
}


\maketitle
\renewcommand{\thefootnote}{\arabic{footnote}}


\section{Introduction}

From 2011 to the present day, a new weakly nonlocal action principle has been proposed and extensively studied in order to overcome the renormalizability issue of the Einstein--Hilbert gravitational theory and the unitarity issue of Stelle gravity \cite{Modesto:2011kw,Biswas:2011ar,Modesto:2017sdr} (see \cite{BasiBeneito:2022wux,Buoninfante:2022ild,Koshelev:2023elc} for reviews). At the classical level, the dynamics is well defined and the Cauchy problem requires a finite number of initial conditions \cite{CMN3}. At the quantum level, the new theory was first introduced by Krasnikov in 1987 \cite{Kra87}, but the first version potentially consistent at the quantum level was due to Kuz'min in 1989 \cite{Kuzmin:1989sp}. Indeed, the exponential nonlocal operators employed in \cite{Kra87} are not compatible with power-counting renormalizability and an asymptotically polynomial nonlocality is needed instead \cite{Kuzmin:1989sp}, giving rise to a class of nonlocal quantum field theories which we may dub asymptotically local. Finiteness at the quantum level for any spacetime dimension was achieved in 2014, when the theory was completed with extra local operators at least cubic in the Ricci tensor \cite{Modesto:2014lga}. Tree-level unitarity was first pointed out in \cite{Kra87}, but proved at any perturbative order via Efimov analytic continuation in \cite{Pius:2016jsl,Briscese:2018oyx}. The latter is needed because Feynman diagrams cannot be computed directly in Lorentzian signature \cite{Briscese:2021mob,Calcagni:2024xku}.

Having at our disposal a candidate for an ultraviolet (UV) complete theory of quantum gravity, we now turn to its cosmology. In this paper, we begin to draw a picture of the early universe without introducing any inflationary era driven by an extra scalar degree of freedom. It deserves to be recalled that a UV completion of the Starobinsky inflationary model in nonlocal gravity has been proposed and extensively studied \cite{Briscese:2012ys,Briscese:2013lna,Koshelev:2016xqb,Koshelev:2017tvv,SravanKumar:2018dlo,Calcagni:2020tvw}. However, a careful look at the perturbative expansion of the dynamics \cite{SravanKumar:2018dlo} suggests a serious tension between renormalizability and stability, since the form factors (nonlocal operators) of that model were designed to guarantee stability on a de Sitter background, but not on others such as, for example, Starobinsky's quasi-de Sitter solution or generic Friedmann--Lema\^{i}tre--Robertson--Walker (FLRW) metrics.

We will see that, abandoning the efforts to embed inflation in nonlocal quantum gravity and just using the symmetries of the theory, the well-known problems of the classical hot-big-bang model, namely, the flatness, the horizon, the monopole, the trans-Planckian, and the singularity problems are elegantly and naturally solved by the Weyl invariance of the theory, regardless of the details of the theory beyond the Planck energy. Indeed, what really matters is not the specific model of quantum gravity, or the type of action principle, but the universality class that collects together all the finite theories at the quantum level. The absence of Weyl anomaly turns out to be an extremely powerful tool to solve the problems of classical cosmology mentioned above. When the symmetry is spontaneously broken and the dilaton takes an expectation value equal to the Planck mass, thermal fluctuations generate a primordial scalar spectrum consistent with the low-multipole cosmic microwave background (CMB) spectrum. As we discuss in a companion paper \cite{Calcagni:2022tuz}, the primordial scalar and tensor spectra are nearly but not exactly scale invariant thanks to logarithmic quantum-gravity corrections to the correlation functions. Once again, here Weyl symmetry is anomaly-free and, thus, it is manifest at the classical as well as at the quantum level, but spontaneously broken at the Planck scale. Therefore, such symmetry offers a natural solution to the above problems regardless of the details of the Lagrangian, and the physics of the cosmos we will discuss in this paper takes a universal character. In particular, the tensor spectrum turns out to have interesting properties which depend only very weakly on the details of the theory: it is blue-tilted and its amplitude is large enough to be detectable by upgrades of present-generation experiments.

Let us mention that Weyl invariance \cite{Wey21,Dirac:1973gk} has already been invoked in the literature as an alternative to inflation \cite{Antoniadis:1996dj,Antoniadis:2011ib,Amelino-Camelia:2013gna,Amelino-Camelia:2015dqa,Agrawal:2020xek}, as an agent acting together with inflation \cite{Wetterich:2019qzx}, in relation with the cosmological constant problem \cite{Wetterich:2019qzx,Wetterich:1987fm} or to cure cosmological \cite{narlikar:1977nf} and black-hole singularities \cite{Prester:2013fia}. These models all view conformal invariance as a fundamental extension of spacetime symmetries and the key to understand Nature at high energies and short scales, a task for which the Lorentz group alone might be insufficient \cite{tHooft:2015vaz}. They also hold that, since the world we observe is not scale invariant, a symmetry breaking mechanism must be implemented, either at early times or at short scales, or both. Attempts to elevate this mechanism from the status of an abstract assumption tend to vary considerably due to theoretical uncertainty. An example of explicit breaking is via a modified dispersion relation \cite{Amelino-Camelia:2015dqa}. In general, however, symmetry breaking is expected to be spontaneous, which can be achieved
when the dilaton \cite{Prester:2013fia,tHooft:2015vaz} or some other scalar \cite{Agrawal:2020xek,Witten:1988xi} acquires an expectation value. Since perturbations are gauge-dependent in the presence of Weyl symmetry, the freezing of the dilaton into an expectation value is tantamount to a gauge fixing \cite{Prester:2013fia,tHooft:2015vaz} and can happen with a trivially flat potential (see the beginning of section~\ref{higsphs}). A perhaps more compelling mechanism, adopted in \cite{Wetterich:2019qzx,Wetterich:1987fm} and in the present paper, is to introduce a conformally invariant interacting potential for the dilaton and another scalar field \cite{Wetterich:1987fm}, possibly the Higgs \cite{Bars:2006dy,Bars:2013yba}.
	
The above models are quite different from the one presented in this paper since, in our case, behind Weyl invariance we have a concrete quantum theory of gravity and of the other fundamental interactions, with a spontaneous symmetry breaking mechanism in place. In turn, this leads to sharper and unexplored predictions in the tensor sector \cite{Calcagni:2022tuz}.

The paper is organized as follows. In section~\ref{NLQGsec}, we review a simple version of nonlocal quantum gravity with matter and the role of finiteness in enforcing Weyl symmetry. We also introduce for the first time the details of an often-used form factor, here simplified further (section~\ref{sec2a}), and discuss the scales characterizing the theory (section~\ref{sec2b}). In section~\ref{Early}, we describe the Weyl trans-Planckian phase (section \ref{Trans-Planckian-conformal-phase}), the Higgs or post-Planckian phase and the one-loop quantum corrections characterizing the latter (section \ref{higsphs}) and the spacetime initial conditions in the trans- and post-Planckian phases (section \ref{incon}). In section~\ref{ProblemsCosmology}, we show how Weyl symmetry solves the problems of the standard hot-big-bang model in the trans-Planckian phase. Cosmological background solutions and their relation with the flatness and horizon problems are presented in section~\ref{sotns}. 
Conclusions are drawn in section~\ref{conclusions}, where future extensions of the results are also discussed. Several appendices contain technical material. 


\section{Asymptotically local quantum gravity}\label{NLQGsec}

A nonlocal gravitational theory coupled to matter has been introduced in \cite{Modesto:2021ief,Modesto:2021okr}, consistently with the following properties: (i) unitarity at any perturbative order in the loop expansion \cite{Modesto:2021ief}, (ii) same linear and nonlinear stability properties and same tree-level scattering amplitudes as for the Einstein--Hilbert theory \cite{Dona:2015tra,Modesto:2021soh}, (iii) same macro-causality properties \cite{GiMo}, and (iv) super-renormalizability or finiteness at the quantum level \cite{Calcagni:2023goc}.\footnote{Other nonlocal models for gauge theories or extending the Standard Model of particle physics were proposed in \cite{MoRa2} and \cite{Modesto:2015foa}, respectively.} Therefore, and although we are only beginning to tap its potential, the proposal in \cite{Modesto:2021ief} has the right requirements to be a unified theory of all fundamental interactions. 

We work in $(-,+,\cdots,+)$ signature and define the reduced Planck mass
\be\label{plama}
\Mpl^2\coloneqq\frac{1}{8\pi G}\,,
\ee
where $G$ is Newton's constant and the energy dimensionality of the Planck mass is $[\Mpl]=(D-2)/2$. 

The full action in $D$ topological dimensions reads
\ba
\hspace{-0.6cm}&& S[\Phi_i] = \int \rmd^D x \sqrt{|g|} \left[\cL_{\rm loc} + E_i \, F^{ij}(\Hes) \, E_j  + V(E_i)\right],\label{action} \\
\hspace{-0.6cm}&& S_{\rm loc} = \int \rmd^D x \sqrt{|g|} \,\cL_{\rm loc}\,,\label{localS} \\
\hspace{-0.6cm}&& \cL_{\rm loc} = \frac{\Mpl^{2}}{2}\,R+\cL_{\rm m}(\Phi_i)\,,\label{localL}\\
\hspace{-0.6cm}&& E_i \coloneqq \frac{\de S_{\rm loc}}{\de \Phi_i(x)}\,,\label{localEoM}\\
\hspace{-0.6cm}&& \Hes_{ki} \coloneqq \frac{\de E_i}{\de\Phi_k} = \frac{\de^2 S_{\rm loc}}{\de\Phi_k\de\Phi_i}\,,\label{Hessian}
\ea
where the ``potential'' $V(E_i)$, whose role will be described in section \ref{ficon}, is a collection of local operators at least cubic in the extremals $E_i$ (with energy dimensionality $[E_i]=D-[\Phi_i]$) and $\Phi_i$ (with canonical dimensionality $[\Phi_i]$) is any field in the theory, including the spacetime metric $g_{\mu\nu}$, the dilaton $\Phi$, fermions $\psi$ and gauge fields $A_\mu$: 
\be
\Phi_i \in \{ g_{\mu\nu}, \Phi, \psi, A_\mu,\,\dots \}\,. \label{fields}
\ee
The functional variations $\de/\de\Phi_i$ contain a weight factor $1/\sqrt{|g|}$, so that $\de\Phi_j(x)/\de\Phi_i(y)=\de_{ij}\de^D(x-y)/\sqrt{|g|}$. For the local Einstein--Hilbert Lagrangian \Eq{localL} with local matter Lagrangian $\cL_{\rm m}$, the local equations of motion for the metric are
\be\label{leom00}
E_{\mu\nu} = \frac12\left(\Mpl^{2} G_{\mu\nu} - T_{\mu\nu}\right) = 0\,,
\ee
where $T_{\mu\nu}\coloneqq-2\de S_{\rm m}/\de g^{\mu\nu}$ is the energy-momentum tensor. These are not the equations of motion of the nonlocal theory, as we will see shortly.

In order to make explicit the hidden Weyl symmetry of the theory \Eq{action}, we make the following replacements in the action:
\be
g_{\mu\nu}\eqqcolon \phi^2 \, \hat{g}_{\mu\nu} \,, \qquad \Phi = \frac{\hat{\Phi}}{{\phi} } \, , \qquad \psi = \frac{\hat{\psi}}{\phi^{\frac{3}{2}}} \,, \qquad A_\mu = \hat{A}_\mu \, .
\label{MatterConfInv}
\ee
All the fields $\hat{g}_{\mu\nu}$, $\hat{\Phi}$, $\hat{\psi}$, and $\hat{A}_\mu$ are rescaled by the dimensionless field $\phi$ in such a way that combined terms take an invariant form under Weyl transformations when $\phi$ is treated like a scalar,\footnote{We fix our terminology as in \cite{Shaposhnikov:2022dou}. A \emph{Weyl transformation} $g^\prime_{\mu\nu} = \Om^2(x)\,g_{\mu\nu}$ connects two generic curved backgrounds $g_{\mu\nu}$ and $g^\prime_{\mu\nu}$ by a non-constant conformal factor $\Om(x)$. A \emph{conformal transformation} is a special case of Weyl transformation where $g_{\mu\nu}=\eta_{\mu\nu}$. A \emph{scale transformation} is a special case of Weyl transformation where $\Om={\rm const}$. Weyl, conformal and scale invariance are defined accordingly.}
\be
\hat{g}^\prime_{\mu\nu} = \Om^2(x)\, \hat{g}_{\mu\nu} \, , \qquad \phi^\prime = \Om^{-1}(x)\, \phi \,,\qquad \dots\,. 
\label{WI}
\ee
Once the fields \Eq{MatterConfInv} are replaced in the action \Eq{action} and after the rescaling $\phi\to\phi/\Mpl^{(D-2)/2}$, so that the dilaton acquires dimensionality $[\phi] = (D-2)/2$, all the coupling constants of the theory become dimensionless. 

Finally, the form factor $F(\Hes)$ is defined in terms of an entire analytic function $\H(\Hes)$ of the Hessian operator $\Hes$,
\be
2 \Hes F(\Hes) \coloneqq \rme^{\bar\H(\Hes)} - 1 \,,\qquad \bar{\H}(\Hes) \coloneqq \H(\Hes)-\H(0)\,.\label{FF}
\ee
Note that $[F^{ij}]=-[\Hes_{ij}]=[\Phi_i]+[\Phi_j]-D$. In accordance with the standard treatment of nonlocal dynamical systems \cite{BasiBeneito:2022wux}, we assume $\H$ to be analytic (more precisely, holomorphic) and entire. If it is analytic, then it admits a series representation, which makes it easier to calculate the equations of motion. This assumption is also dictated by the fact that nonanalytic operators can be associated with infrared modifications of gravity (e.g., \cite{Barvinsky:2003kg,Deser:2007jk,Zhang:2016ykx,Belgacem:2020pdz} and references in \cite{Calcagni:2021hve}) that have little to do with a fundamental theory with new UV physics. On the other hand, the condition on $\H$ to be entire, invoked since the early days of nonlocal form factors in dynamical systems and quantum field theory \cite{PU,Efimov:1967dpd,Efimov:1967pjn} and quantum gravity \cite{Modesto:2011kw,Biswas:2011ar,Kra87,Kuzmin:1989sp,Tom97,Biswas:2005qr}, is to avoid modifications of the perturbative spectrum and, in particular, to introduce extra ghost degrees of freedom.

Thanks to \Eqq{FF}, the gravitational equations of motions of the nonlocal theory \Eq{action} can be written as
\be
\big[\rme^{\bar\H(\Hes)}\big]_{\mu\nu}^{\s\t} \, E_{\s\t} +O(E_{\mu\nu}^2)=0\,.\label{LEOM0}
\ee
Since, by construction, the form factor \Eq{FF} does not possess any poles or zeros in the whole complex plane at finite distance, any solution of the local theory \Eq{leom00} is also a solution of \Eqq{LEOM0}. The fact that we have more background solutions than Einstein's gravity does not imply that the theory is less predictive, since physical observables turn out to depend on very few free parameters \cite{Calcagni:2022tuz}.

The function $\H$ actually depends on the dimensionless quantity $z\propto\Hes/(\Mpl^2\mst^2)$, where $\mst$ is an energy scale introduced to compute the dimensionality of the operator $\Hes$.\footnote{The cosmologist should not confuse $z$ with the redshift (which we will never use in this paper) or $\H$ with the Hubble parameter $H$.} However, according to the discussion above, the theory \Eq{action} is Weyl invariant and the scale $\mst$ can be reabsorbed in the dilaton field $\phi$. With a slight but innocuous abuse of notation, we will freely interchange $\H(\Hes)$ with $\H(z)$ without further notice. The definition of $\H$ is
\ba
\hspace{-0.5cm} \H(z) = \int_0^{p(z)} \rmd w \, \frac{1 - \rme^{-w}}{w} = \g_\textsc{e} + \G[0, p(z)] + \ln p(z) \,,\label{H}
\ea
where $\g_\textsc{e}\approx 0.577$ is the Euler--Mascheroni constant, $\G$ is the upper incomplete gamma function, and $p(z)$ is a polynomial of degree $n+1$ in the variable $z$:
\be
p(z) = a_0 + a_1 z + a_2 z^2 + \dots + a_{n+1} z^{n+1}\,, \qquad a_i \in \mathbb{R}\,. \label{Poly}
\ee

The form factor \Eq{H} is actually a known special function, dubbed ${\rm Ein}(z)$ and called complementary exponential integral \cite[formula 6.2.3]{NIST}. The logarithmic divergence at $z=0$ is exactly canceled out by the divergence of $\G(0,z)$ and, indeed, $\H$ is finite and real on the whole real axis. Likewise, its exponential
\be
\rme^{\H(z)}= \rme^{\g_\textsc{e} + \G[0,p(z)]} \,  p(z)\label{MCFormFactor}
\ee
is well-defined and positive definite everywhere on the real axis (figure~\ref{fig1}). Since $\exp\H>0$, propagators do not change sign and there is no unitarity issue.
\begin{figure}
\bc
\hspace{-0.5cm}
\includegraphics[width=8.5cm]{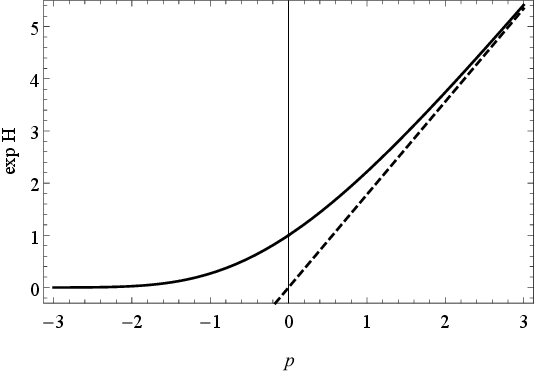}
\caption{Form factor \Eq{MCFormFactor} (solid curve) and its asymptotic limit \Eq{UVlimit}.\label{fig1}}
\ec
\end{figure}

The only constraint we place on the polynomial is that its real part be positive definite when $z\to\pm\infty$:
\be\label{unicacondi}
\Re\, p(z\to\pm\infty) > 0\,.
\ee
In this case, in the UV limit ($z\gg 1$)
\be
\rme^{\H(z)}\stackrel{\textsc{uv}}{\simeq} \rme^{\g_\textsc{e}} p(z)\,,\label{UVlimit}
\ee
which explains the name \emph{asymptotically polynomial} usually given to this class of form factors. On the other hand, in the infrared (IR; $z\ll 1$), the analyticity of the form factor provides an effective expansion of the action in higher-derivative operators. 

The form factor \Eq{H} includes those proposed by Kuz'min \cite{Kuzmin:1989sp}, Tomboulis \cite{Tom97}, and Modesto \cite{Modesto:2011kw}, but it is presented here in a simpler way, with fewer parameters and a more minimalistic constraint \Eq{unicacondi}.


\subsection{Choice of form factor}\label{sec2a}

For the sake of simplicity, we will work with a quartic polynomial:
\be
p(z) = a_0 + a_1 z + a_4 z^4  \, , \qquad a_0, \, a_1, \, a_4\in \mathbb{R} \,.
\label{PolySimple}
\ee
Other quartic polynomials are possible, e.g., $p(z)=(a_0 + a_1 z + a_2 z^2)^2$, but \Eqq{PolySimple} will suffice. This simple expression suits the minimal super-renormalizable (or finite) theory in $D=4$ dimensions, for which the polynomial has degree four. We removed the monomials $z^2$ and $z^3$ since they only affect the transient regime between the UV and the IR and, therefore, cannot change the results in this paper qualitatively.

In order to satisfy the condition \Eq{unicacondi}, $a_4>0$. We fix this coefficient to $a_4=1$ noting that the mass scale $\mst$ in $z$ can be redefined as $\Lst\coloneqq\mst/a_4^{1/4}$. For later convenience, we redefine $a_0\eqqcolon b_0$ and $b\coloneqq -a_1>0$, where the sign of $b$ is chosen on physical grounds that will become clear later. The polynomial \Eq{PolySimple} is rewritten as
\be\label{PolySimple2}
p(z) = b_0-b\,z+z^4\,.
\ee
A most important property we would like to obtain is to have an intermediate energy regime (still in the UV) where the linear term in \Eqq{PolySimple2} dominates over the other two. In the UV limit, this happens when
\ba
\frac{\rme^{\H(z)}}{\rme^{\H(0)}} &\stackrel{\textsc{uv}}{\simeq}& \rme^{\tilde{\g}_\textsc{e}}p(z)\stackrel{\textrm{\tiny \Eq{PolySimple2}}}{=} \rme^{\tilde{\g}_\textsc{e}}(b_0 - b\, z + z^4)\nn
&\simeq& \rme^{\tilde{\g}_\textsc{e}}(-b\, z)\,,\label{bnew}\\
\tilde{\g}_\textsc{e} &\coloneqq& \g_\textsc{e}-\H(0)\,,\label{tildeg}
\ea
which holds for $b_0 \ll b\,|z|$ and $b\,|z| \gg |z^4|$. For consistency, when $b_0<1$, the first inequality must be supplemented by the condition that we are still in the asymptotically polynomial UV regime \Eq{UVlimit}, i.e.,
\be\label{interint}
\frac{{\rm max}(1,b_0)}{b} \ll |z| \ll b^\frac{1}{3}.
\ee
This condition is satisfied, for instance, when $1<b_0\ll b$, or when $b_0=0$, $b\gtrsim O(1)$, and $z\gtrsim O(1)$. In fact, one can even relax the strong inequalities in \Eq{interint} and consider values of $z$ near the upper and lower bounds. In the worst case, $b_0=O(b z)=z^4$ and \Eqq{bnew} would acquire a factor of 3. However, this factor is not very important because, as seen in \cite{Calcagni:2022tuz}, $b$ should be large in order to satisfy CMB observations. In other words, any breaking of the approximation \Eq{bnew} would not lead to any inconsistency but it would alter the relation between the theoretical value and the experimental value of the free parameter $b$, without changing its order of magnitude. We comment on physical explanations of large values of $b$ in \cite{Calcagni:2022tuz}.

\subsection{Scales and parameters}\label{sec2b}

Let us now comment on the energy scale $\Lst$ in the polynomial. More generally, a sort of naturalness principle would require the theory to feature multiple scales through the coefficients in the polynomial \Eq{Poly}. Moreover, \emph{a priori} all the monomials in \Eqq{Poly} should be present. In order to make the above scale dependence explicit, we can rewrite \Eqq{Poly} in the form 
\ba
p(\B) = b_0 + b_1\frac{\B}{\Lambda^2_1}+ \dots +\left(b_n \, \frac{\B}{\Lambda^2_n} \right)^n \,, \qquad
\Lambda_1 \ll  \Lambda_2 \ll \dots  \ll \Lambda_n \,,\label{PolyMS}
\ea
where each power of the Laplace--Beltrami operator $\B$ contributes at a different scale and $b_i=O(1)$ are dimensionless positive or negative constants. As said above, the theory is Weyl invariant, so that all the above scales $\Lambda_i$ are proportional to the expectation value of the dilaton field multiplied by dimensionless constants related to the coefficients $a_i$ in \Eqq{Poly}, which can take any value in order to introduce a hierarchy of scales. In \Eqq{PolyMS}, the $\B$ operator can be replaced by the Hessian $\Hes$ without affecting the arguments presented here. 

However, in $D=4$ dimensions all the physical implications of the theory in the trans-Planckian as well as in the post-Planckian regime are captured by the simple polynomial \Eq{PolySimple2}, which we can rewrite as
\ban
p(\B) = b_0 +  b_1 \, \frac{\B}{\Lambda^2_1} + \left(b_4 \, \frac{\B}{\Lambda_4^2} \right)^{4}
\eqqcolon b_0 - b \, \frac{\B}{\Lambda^2_0}+\left(\frac{\B}{\Lst^2}\right)^{4}, 
\ean
where now $\Lambda_0$ and $\Lst$ can be of the same order of magnitude and $b$ can be large. 
 As shown in \cite{Calcagni:2022tuz}, in $D=4$ dimensions the quartic term is responsible for exact scale invariance (Harrison--Zel'dovich spectra), while the observed quasi-scale-invariant scalar spectrum requires the presence of the monomial linear in $\B$ and with negative coefficient.

We see that the theory \Eq{action} contains three scales: $\Mpl$, $\Lambda_0$, and $\Lst$. However, if $b \gg 1$ we can make the identification $\Lambda_0 = \Lst$ without any loss of generality. Indeed, $b$ simply defines the ratio between $\Lambda_0$ and $\Lst$:
\be
\boxd{p(\B) = b_0 - b \, \frac{\B}{\Lst^2}+\left(\frac{\B}{\Lst^2}\right)^{4}.}
\label{PolySimple20}
\ee
In the case of gravity, one should make the replacement $\B\to 4\Hes/\Mpl^2$:
\be
p(\Hes) = b_0 - b \, \frac{4\Hes}{\Mpl^2\Lst^2}+\left(\frac{4\Hes}{\Mpl^2\Lst^2}\right)^4,
\label{PolySimple3}
\ee
where $[\Hes]=4$. Therefore, the most general theory has two mass scales and one dimensionless free parameter:
\be
\Mpl \,, \quad \Lst\,, \quad b\,.
\label{3scales}
\ee
We discuss the values of $\Lst$ and $b$ in \cite{Calcagni:2022tuz} but, here, we leave them unspecified.

We can gain further understanding of the role of each term in the action in $D>4$ dimensions. After the spontaneous symmetry breaking of Weyl invariance, the effective theory amounts to a collection of higher-derivative operators characterizing the following extra-dimensional renormalizable generalization of Stelle gravity \cite{Ste77,Ste78}:
\ba
\hspace{-0.2cm}
 \cL \sim 
R + \frac{\mathcal{R}^2}{\Lambda_0^2}  +  \mathcal{R} \frac{\B}{\Lambda_2^4} \mathcal{R} + \dots + \mathcal{R} \frac{\B^{ \frac{D-4}{2} }}{(\Lambda_{D-3})^{D-4}}  \mathcal{R} 
\,,\label{MultiStelle}
\ea 
where $\mathcal{R}$ stands for $R, R_{\mu\nu}, R_{\mu\nu\rho\s}$, while the ellipsis includes $O(\mathcal{R}^3)$ operators which appear when $D>4$. When $D=4$, only the first two terms contribute with, respectively, two and four derivatives, as in \Eqqs{PolySimple20} and \Eq{PolySimple3}. Therefore, the observed quasi-scale-invariant spectrum can be traced back to the renormalizable local theory hidden inside the finite theory \Eq{action}. It deserves to be noticed that, out of all the operators in \Eqq{MultiStelle}, only the one of dimension $D$ matters for reproducing a scale-invariant spectrum, which essentially stems from a propagator $\sim 1/k^D$. Operators of lower dimensionality are relevant in the IR regime in order to reproduce Newton's potential, while operators of higher dimensionality are crucial for the finiteness of the theory and Weyl invariance. A small deviation from exact scale invariance will come from quantum logarithmic corrections.

We conclude this subsection by noting that the theory is defined by a finite number of parameters. This is one of the reasons why it is under control and predictive both at the classical and the quantum level despite having an infinite number of derivatives.


\subsection{Finiteness and Weyl invariance}\label{ficon}

The potential $V(E_{\mu\nu})$ in the action \Eq{action} is needed to guarantee the finiteness of the theory in even dimensions where the presence of the nonlocal form factor is sufficient to make any loop amplitude finite except at the one-loop level. On the other hand, in odd dimensions the theory has no one-loop divergences as a particular feature of the dimensional regularization scheme. All the details on renormalization and finiteness are provided in \cite{Modesto:2014lga} for the vacuum theory and in \cite{Calcagni:2023goc} with matter fields, while perturbative unitarity is treated in \cite{Briscese:2018oyx}. Notice that, due to the presence of the dilaton $\phi$, the theory does not contain any mass scale and, hence, \Eqq{action} is Weyl invariant classically \cite{Modesto:2021ief}. Moreover, since the theory is finite at the quantum level, the beta functions vanish and Weyl symmetry is anomaly free. 

Let us recall that, if a classical field theory is conformally invariant, then this symmetry is broken by quantum corrections if the beta functions of the theory are nonvanishing, while it is preserved is all beta functions are zero. This is easy to see in dimensional regularization at the level of gravity-dilaton operators \cite[section 6.5]{Calcagni:2023goc} but a more elegant way to express the breaking is to calculate the trace $T_\mu^{\ \mu}$ of the energy-momentum tensor. If this trace is nonzero at some loop order, then conformal invariance is broken and we are in the presence of a conformal or Weyl anomaly \cite{Capper:1974ic,Capper:1975ig,Deser:1976yx,Duff:1977ay,Alvarez-Gaume:1983ihn,Antoniadis:1984kd,Deser:1993yx,Duff:1993wm,Mazur:2001aa,Asorey:2003uf,Mottola:2006ew,Boulanger:2007ab,Komargodski:2011vj,Codello:2012sn,Mottola:2016mpl}. The trace anomaly can be calculated from the general counterterms of any theory with infinities. A widely used example is Einstein gravity plus gauge and fermion fields, but it has also been considered in supersymmetric quantum field theory and supergravity \cite{Alvarez-Gaume:1983ihn,Sohnius:1981sn,Fradkin:1983tg,Fradkin:1985am}. The knowledge of the beta functions fixes $T_\mu^{\ \mu}$ completely and \emph{vice versa} \cite{Duff:1977ay}. In effective field theory, one does not have such knowledge because the underlying theory is assumed to be unknown and the beta functions are taken to be generic coefficients.

In contrast, in a finite theory there are no infinities and no counterterms, all beta functions are zero by calculation, and there is no conformal anomaly. This property has been explored in conformal cosmological models as a proof of concept, i.e., without specifying a concrete finite quantum theory behind \cite{Wetterich:2019qzx,Wetterich:1987fm}. Here we advocate the finite version of asymptotically local quantum gravity as an explicit framework in which the vanishing of the beta functions and of the conformal anomaly occurs \cite{Calcagni:2023goc}. At low energies with respect to the scale $\Lst$, the asymptotically polynomial form factors reduce to powers of derivatives, the theory looks like a higher-derivative model, and nonlocal operators become like higher-derivative regulators \cite{Asorey:1996hz}. Therefore, in this limit, all results at any loop order must be in agreement with calculations in the corresponding local theory in the limit $E/\Lambda_*\to 0$, in all sectors. The quantum effective action will have quantum corrections to the local theory, as we will see below. Therefore, by construction, in the low-energy limit of the theory one recovers both the classical and the quantum Standard Model together with its experimentally verified predictions. In particular, the vanishing of the beta functions in the deep UV of the finite version of the theory does not contradict the effective presence of beta functions at the electrodynamics scale, as found, e.g., in precision measurements of the running of the effective fine-structure constant \cite{L3:2000hbp}. At much higher energies, higher-derivative corrections suppressed by powers of $\Lst$ would modify these quantum corrections, while at yet higher energies $O(\Lst)$ nonlocality would ideally make its appearance in accelerators.

At the quantum level, one integrates both the dilaton $\phi$ and the metric $\hat g_{\mu\nu}$ in the path integral. The functional integration on the dilaton gives $Z_0=\int [\cD \phi] [\cD \hat g_{\mu\nu}]\,\exp(\rmi S[\hat g_{\mu\nu}, \phi]) = \int[\cD \hat g_{\mu\nu}] \exp(\rmi S[\hat g_{\mu\nu}, \Lst])$, which is equivalent to perform the gauge fixing $\phi=1$. This makes the energy scale $\Lst$ appearing in the form factors a meaningful reference to define the UV/IR divide $|k^0|=\Lst$. In the Weyl phase, this divide exists only in the mathematical sense specified above and renormalizability or finiteness of a conformal theory are a UV property of the theory only with this meaning, with no reference to physical short scales. When Weyl symmetry is broken or when one considers an effective field theory, the gauge fixing becomes physical because the dilaton takes an expectation value and $\Lst$ becomes a physical scale. At that point, one can appreciate, for instance, that effects of the conformal anomaly are not confined to the UV but can carry through to long distances and leave an imprint in the IR, essentially mediated by light or massless particles \cite{Mottola:2016mpl}.


\section{Gravity at high energy}\label{Early}

It is natural to assume that the description of a Universe governed by any fundamental quantum theory of gravitation should begin at the highest energies with all symmetries unbroken. In the present case, at high energies the theory is Weyl invariant. Then, the evolution of the fields, \emph{in primis} the dilaton, leads to a spontaneous breaking of this symmetry around the Planck scale, a moment marked by the gradual recovery of standard Einstein's gravity. Following this chain of events turns out to provide a viable alternative to inflation.

According to this scenario, the physics of the early Universe consists of a Weyl-invariant trans-Planckian phase and a Higgs Planckian phase in which Weyl symmetry is spontaneously broken. The transition between the two phases is driven by the dilaton field $\phi$ in \Eqq{MatterConfInv}, which we can rewrite as
\be
\phi \equiv \frac{\Mpl}{\sqrt{2}} + \varphi \,,
\label{MPphi}
\ee
where $\varphi$ is the Goldstone boson associated with the spontaneous symmetry breaking of Weyl invariance. For $\vp \gg \Mpl$, $\phi\simeq\varphi$ and the action is Weyl invariant under the transformation \Eq{WI}, while for $\vp \ll \Mpl$, $\phi \simeq \Mpl/\sqrt{2}$ and Weyl symmetry is spontaneously broken by the vacuum solution. Since the Goldstone boson $\vp$ is unobservable, the inequalities $\vp \gg \Mpl$ and $\vp \ll \Mpl$ do not correspond to any regime that could be probed by an ideal Planck-energy detector. Later on, we will characterize the Weyl and the Higgs phase with physical inequalities.


\subsection{Trans-Planckian Weyl phase}\label{Trans-Planckian-conformal-phase}

The trans-Planckian regime is described by a Weyl invariant phase. The theory is manifestly Weyl invariant, the dilaton $\phi$ has an exactly constant potential $U(\phi)=0$, and the correlation functions of all fields are identically zero as a consequence of Weyl symmetry. Indeed, consider an arbitrary set of operators $O_i(x)$ that transform under Weyl symmetry with conformal weight $\Delta_i$, 
\be\label{OmegaDelta}
O_i'(x)=\Om^{\Delta_i}(x)\,O_i(x)\,.
\ee
The vacuum Lorentzian path integral \cite{Calcagni:2024xku} $$Z_0\coloneqq\int [\cD O]\,\exp(\rmi S[O])$$ is invariant under a Weyl transformation, since both the action $S[O_1,O_2,\dots]$ and the functional measure $[\cD O]=\prod_i [\cD O_i]$ are invariant by construction if the quantum symmetry is preserved (no conformal anomaly): $S[O']=S[O]$, $[\cD O']=[\cD O]$. Define the $n$-point correlation function as
\ba
\langle O_1(x_1)  \dots  O_n(x_n) \rangle \coloneqq \frac{1}{Z_0}\int [\cD O]\,O_1(x_1)  \dots  O_n(x_n)\,\rme^{\rmi S[O]},
\ea
so that $Z_0=\langle1\rangle$. Renaming first the dummy integration variables $O_i\to O_i'$, then performing a change of variables via the conformal transformation \Eq{OmegaDelta}, and finally using the property of Weyl invariance for the integration measure and the action, one has \cite[section~2.4.3]{DiFrancesco:1997nk}
\ba
\langle O_1(x_1) \dots O_n(x_n) \rangle
&\stackrel{\textrm{\tiny $O_i\to O_i'$}}{=}& \frac{1}{Z_0}\int [\cD O']\,O_1'(x_1)  \dots  O_n'(x_n)\,\rme^{\rmi S[O']}\nn
&\stackrel{\textrm{\tiny $O_i'=\Om^{\Delta_i}O_i$}}{=}& \frac{1}{Z_0}\int [\cD O']\,O_1'(x_1)  \dots  O_n'(x_n)\,\rme^{\rmi S[O']}\nn
&=&\frac{1}{Z_0}\int [\cD O]\,O_1'(x_1)  \dots  O_n'(x_n)\,\rme^{\rmi S[O]}\nn
&=& \langle O_1'(x_1)  \dots  O_n'(x_n) \rangle\nn
&=& \Om^{\Delta_1}(x_1) \dots \Om^{\Delta_n}(x_n) \langle O_1(x_1) \dots O_n(x_n) \rangle\,,\label{eq31}
\ea
where in the last line we used the fact that $\Om(x)\neq \Om[O(x)]$ is a function of the coordinates but not of the fields.
Since $\Om$ is not a constant, \Eqq{eq31} implies
\be
\langle O_1(x_1) \dots O_n(x_n) \rangle = 0\quad \mbox{for} \quad \Delta_1 \neq 0 \, , \dots, \, \Delta_n \neq 0 \,.
\ee
According to \Eqq{WI}, for the metric and the graviton $\Delta_{\bm{\hat g}, \bm{\hat h}}=2$, while for the dilaton $\Delta_\phi = -1$. Here, boldface symbols stand for rank-2 tensors. On the other hand, correlation functions involving only fields with $\Delta_i=0$ (for example, photons or gauge bosons) can be nonzero. 

Notice that quantum corrections do not affect the propagator in the Weyl phase because the logarithmic operators in \Eq{LQ} are of the form
\be
\ln(-{\rm f}^{-2}\B_{\bm g})\big|_{\bm{g}=\phi^2 \bm{\hat g}} \,,\qquad {\rm f} =\frac{\sqrt{2}\Lst}{\Mpl}\,,
\ee
and they appear together with the dilaton in the Weyl phase. Therefore, operators like 
\ben
\cR({\bm g}) \ln (- {\rm f}^{-2} \B_{\bm g})\, \cR({\bm g})\Big|_{\bm{g}=\phi^2 \bm{\hat g}} 
\een
contribute to three-point functions (three gravitons or two gravitons and a dilaton) but are not involved in perturbative two-point functions. Only when the Weyl symmetry is broken spontaneously in the Higgs phase do quantum corrections quadratic in the curvature tensors take part in the graviton two-point function (see next subsection).

How high can energies be in this phase? As explained in \Eqq{3scales}, in this theory there are in total three scales: $\Mpl$, $\Lst/\sqrt{b}$, and $\Lst$. $\Lst$ is a free scale. $\Mpl$ is another scale independent of $\Lst$ and related to the expectation value of the dilaton field $\langle\phi\rangle=\Mpl/\sqrt{2}$, which is uniquely fixed after Weyl symmetry is spontaneously broken and is given by the measured Newton's constant $G=(8\pi\Mpl^2)^{-1}$ when reading it off from the nonminimally coupled action, obtained from \Eqq{localL} after the transformation \Eq{MatterConfInv} ($\Mpl^2 R/2\leftrightarrow\phi^2 R$). However, in the Weyl phase 
no scale larger than the Planck mass can be probed because, in order to observe such scale, we would need a test particle (or a classical perturbation) having a wavelength $\la$ shorter than the scale at which the Weyl symmetry is spontaneously broken, in this case $\lp\coloneqq\Mpl^{-1}$. However, in the Weyl phase, when $\la<\lp$ we can always make it larger by means of a Weyl rescaling. In other words, any wavelength shorter than the Planck length can be stretched to the Planck length or above with a Weyl transformation. This argument is qualitative and does not fix an absolute upper bound on the energy. Nevertheless, taking as a reference the expectation value of $\phi$ at the scale of symmetry breaking, we see that this absolute upper bound should be of order of $\Mpl$:
\be\label{LambdaMinMP}
\Lst \leq \Mpl\,,\qquad \lst\geq\lp\,,
\ee
where $\lst\coloneqq \Lst^{-1}$. A weaker relation outside the Weyl phase can be obtained consistently with perturbative quantum field theory \cite{Calcagni:2022tuz}. In this paper, we will remain as general as possible and fix $\Lst$ at or near the Planck mass only in \cite{Calcagni:2022tuz}.


\subsection{Post-Planckian Higgs phase}\label{higsphs}

The post-Planckian phase takes place at energy scales below the Planck mass $\Mpl$, when the magnitude of the Goldstone boson is much smaller than $\Mpl$. In this phase, the classical action is obtained simply replacing \Eqq{MPphi} into \Eq{action} for $\vp \ll \Mpl$ or, equivalently, imposing the unitary gauge to fix $\vp=0$.

In standard quantum field theory, whenever we choose an exact solution we always have to check its stability against perturbations. Usually, in the presence of, say, gauge symmetries, a potential with local minima is required to ensure the stability of a given vacuum choice, lest quantum fluctuations displace the field from there. This setting is unavoidable when field perturbations are physical or, in other words, gauge invariant. However, in a Weyl-invariant theory the dilaton is pure gauge: It can always be fixed to a value $\phi={\rm const}$ (where ${\rm const}=\Mpl$ acquires physical meaning only after breaking the symmetry) and perturbations $\de\phi=\phi-\Mpl$ are gauge-dependent. Therefore, we neither need a potential nor to expand around a minimum. Also, the perturbation $h_{\mu\nu}$ around a background with $\phi=\Mpl+\de\phi$ is conformally equivalent to a perturbation $h_{\mu\nu}$ on a background with $\phi=\Mpl$ \cite{Percacci:2011uf}. Thus, if one perturbs a $\phi={\rm const}$ and $g_{\mu\nu}=g_{\mu\nu}^{(0)}$ background such that $R_{\mu\nu}[g_{\mu\nu}^{(0)}]=0$, this vacuum turns out to be stable against spin-2 perturbations, while there are no spin-0 perturbations because they are gauge-dependent.

This peculiarity of Weyl-invariant theories makes the mechanism of spontaneous symmetry breaking different with respect to the usual Higgs mechanism valid in other theories. In particular, spontaneous symmetry breaking of Weyl invariance happens when the dilaton picks a given expectation value (or, in other words, becomes purely classical), even if its potential is zero \cite{tHooft:2015vaz,Antoniadis:1984kd,Fradkin:1985am,Englert:1976ep,tHooft:2011aa,Ilg12,Bars:2013yba}. Our world is built on a vacuum with a certain value among all the others, a state not invariant under Weyl symmetry. In this scenario, the dilaton potential is flat and there are an infinite number of vacua. Symmetry breaking occurs when one such vacuum is selected and the field takes the expectation value $\langle\phi\rangle=\Mpl/\sqrt{2}$. This specific value, which has no special meaning during the Weyl phase, is determined \emph{a posteriori} by experiments.\footnote{Sometimes it has been claimed that setting $\phi=\phi_0={\rm const}$ in conformal gravity might not correspond to a spontaneous breaking of Weyl symmetry because one can perform a conformal transformation $\Om=(\phi/\phi_0)^w$ (where $w$ is a constant) such that the Lagrangian with $\phi=\phi_0$ is transformed back to a conformally invariant Lagrangian \cite{Lasenby:2015dba,Hobson:2022ahe}. According to this reasoning, $\phi=\phi_0$ would be just a gauge choice not affecting the symmetry of the system in any meaningful way. Although this is mathematically correct, it misses the elementary particle-physics point that, by definition, spontaneous symmetry breaking entails a vacuum choice after which no symmetry transformation is possible. A classic example is the choice of vacuum in the ${\rm U}(1)$-invariant set of vacua of the Mexican-hat potential for the Higgs field $\phi=\Phi\exp(\rmi\theta)$, corresponding to a choice of phase $\theta$. More precisely, the total Lagrangian does retain the symmetry but such symmetry is broken for perturbations around the vacuum. The argument of \cite{Lasenby:2015dba,Hobson:2022ahe} does not take these perturbations $\vp$ of the dilaton into account.}

Assuming without further details that the dilaton or another scalar field \cite{Witten:1988xi} picks a given expectation value may be somewhat disturbing. For this reason, we can outline other ways to break Weyl symmetry as follows.

The closest mechanism to the usual Higgs one relies on a dilaton potential with absolute minima corresponding to dynamically preferred states. This potential should be built into the action \Eq{action}, in such a way as to respect Weyl invariance and renormalizability properties and to admit static nonzero minima for the dilaton:
\be\label{minima}
\Phi_i^{\rm min}\in\{\phi={\rm const}\neq 0,\,\hat g_{\mu\nu}(x),\,\hat\Phi_i^{\rm SM}(x),\,\dots\}\,,
\ee
where $\hat g_{\mu\nu}$ is the metric from which one extracts the graviton. $\hat\Phi_i^{\rm SM}(x)$ are the matter fields of the Standard Model, and $\dots$ are other matter fields, which we ignore from now on. The simplest option is to have a nontrivial dilaton potential from the minimalistic construction of section \ref{NLQGsec}, where the dilaton arises from the field redefinition \Eq{MatterConfInv}, $g_{\mu\nu}=\phi^2\hat g_{\mu\nu}$. In four dimensions, one has a nonlocal Lagrangian of the form 
\ba
\sqrt{|g|}\cL[\phi,g_{\mu\nu},\Phi_i^{\rm SM}]&=&\sqrt{|\hat g|}\phi^4\cL[\phi^2\hat g_{\mu\nu},\phi^{-\Delta_i}\hat\Phi_i^{\rm SM}]\nn
&=&\sqrt{|\hat g|}\hat\cL\,,
\ea
where $\Delta_i$ is the conformal weight of matter fields and $\cL$ is obtained from the generating local Lagrangian \Eq{localL}:
\be\label{localL2}
\cL_{\rm loc}=\cL_{\rm loc}[\phi^2\hat g_{\mu\nu},\phi^{-\Delta_i}\hat\Phi_i^{\rm SM}]\,.
\ee
Since $\sqrt{|g|}=\sqrt{|\hat g|}\phi^4$, a cosmological constant term $\Lambda_{\rm cc}$ in the generating local Lagrangian \Eq{localL} would immediately yield a quartic potential $\phi^4$, which would be a good starting point were it not for the presence of other terms where $\phi$ is nonminimally coupled with the metric and matter fields. The task would be to check that there is a minimum such as \Eq{minima} where the metric and all the fields of the Standard Model remain dynamical.

Another possibility is to modify the minimalistic approach  of section \ref{NLQGsec} and construct a nonlocal action \Eq{action} directly with an explicit dilaton dependence, which can be achieved by adding Weyl-invariant extra terms $\cL^{\rm loc}_2$ in the local Lagrangian \Eq{localL2}:
\be
\cL_{\rm loc}=\cL^{\rm loc}_1[\phi^2\hat g_{\mu\nu},\phi^{-\Delta_i}\hat\Phi_i^{\rm SM}]+\cL^{\rm loc}_2[\phi,\hat g_{\mu\nu},\hat\Phi_i^{\rm SM}]\,.
\ee
Some proposals in the literature for an effective dilaton potential might prove useful also in this theory, after adaptations \cite{Bars:2006dy,Bars:2013yba,Schwimmer:2010za,Kubo:2022jwu}. Take for instance the Higgs-dilaton potential \cite{Bars:2006dy,Bars:2013yba}
\be\label{hidil}
\cL^{\rm loc}_2[\phi,\mathfrak{h}]=\la (\mathfrak{h}^\dagger\mathfrak{h}-\a\phi^2)^2+\la'\phi^4\,,
\ee
where $\mathfrak{h}$ is the Higgs field and $\la$, $\la'$, and $\a$ are dimensionless, positive constant couplings. The term $\sqrt{|\hat g|}\cL^{\rm loc}_2$ is Weyl invariant and has the attractive property that its minima are the locus $\mathfrak{h}^\dagger\mathfrak{h}=\a\phi^2$. In particular, $|\mathfrak{h}|={\rm const}$ when $\phi={\rm const}$, so that, when the vacuum expectation value of the Higgs gives the observed masses of the gauge bosons of the Standard Model, then the dilaton remains constant in order to stay in the trough, and \emph{vice versa}. Again, an analysis of the equations of motion is required to have a robust picture of this mechanism, but a great advantage is that the solutions of the equations of motions of the generating local system are solutions of the nonlocal theory, with the same stability properties. Therefore, the dynamical analysis of \cite{Bars:2006dy,Bars:2013yba} holds. Of course, there may be other solutions in the nonlocal case.

The Higgs-dilaton symmetry breaking mechanism based on \Eq{hidil} \cite{Bars:2006dy,Bars:2013yba} has the advantage of giving physical meaning to the inequalities $\vp \gg \Mpl$ and $\vp \ll \Mpl$ defining, respectively, the trans-Planckian and the post-Planckian phase. In fact, in the trans-Planckian phase and near the minimum trough of the potential \Eq{hidil} one has
\be
\mathfrak{h}^\dagger\mathfrak{h}=\a \phi^2 \sim\a \vp^2\gg \a\Mpl^2\qquad\Longrightarrow\qquad \sqrt{ \mathfrak{h}^\dagger\mathfrak{h}}\gg \sqrt{\a}\Mpl\,,
\ee
while in the post-Planckian phase $\vp \ll \Mpl$ we have
\be
\mathfrak{h}^\dagger\mathfrak{h}=\a\phi^2 \simeq \frac{\a}{2}\Mpl^2\qquad\Longrightarrow\qquad \sqrt{ \mathfrak{h}^\dagger\mathfrak{h}}\simeq \sqrt{\frac{\a}{2}}\Mpl\,.
\ee
These relations now involve a physical particle rather than a Goldstone boson.


To get the dynamics with quantum corrections, we work in perturbation theory. Nonlocal quantum gravity is mainly studied as a perturbative field theory where the graviton is a perturbation around a given background. This procedure, common in metric perturbative approaches to quantum gravity, is to be considered legitimate at all scales (even around the Planck scale) as long as the perturbative expansion is self-consistent and quantum corrections are kept subdominant order by order. The cosmological model derived in this paper and in \cite{Calcagni:2022tuz} respects this constraint by construction.

According to the results of section \ref{sec2a}, we can ignore the constant term in the form factor when the energy $E\coloneqq k^0\gtrsim \Est$, where 
\be\label{Estar}
\Est\coloneqq   \sqrt{\frac{{\rm max}(1,b_0)}{b}}\,\Lst\,.
\ee
At this scale, the action \Eq{action} reduces to a particular local theory quadratic in the curvature which, in vacuum, is Stelle gravity with a special choice of coefficients (appendix \ref{appeA1}). The subscript ``hd'' is a reminder that we are working in this higher-derivative limit. Then, one can show that, in the intermediate UV regime
\be\label{interm}
\Est\lesssim E\lesssim\Lst\,,
\ee
(which comes from \Eqq{interint} where the left inequality is weak for the reasons argued above, while the right inequality is weak because $b\gtrsim O(1)$) the nonlocal Lagrangian in \Eqq{action} acquires logarithmic quantum corrections at one loop (appendix \ref{appeA2}):
\bs\label{QEA}\ba 
\cL_{\Est} &=& \frac{\Mpl^2}{2} R + \cL_{\rm m} + E_{\mu\nu} E^{\mu\nu}  + V(E_{\mu\nu}) + \cL_{Q} \,,\\
\cL_{Q} &=& \b_R R \ln \left(-\frac{\B}{\Lst^2\de_0^2}\right) R + 
\b_{\rm Ric} R_{\mu\nu}  \ln \left(-\frac{\B}{\Lst^2}\right) R^{\mu\nu},\label{LQ}
\ea\es
which are exactly the same quantum corrections of Stelle's theory upon the identification of 
the cut-off scale $\Lambda_\textsc{uv}$ of Stelle quantum gravity with the fundamental scale of the nonlocal theory, $\Lambda_\textsc{uv} \propto\Lst$. Indeed, the higher-derivative operators (depending both on the potential $V(E_{\mu\nu})$ and on the leading terms in the polynomial $p(z)$) can be seen as regulators for quadratic gravity, and, in the limit $E \ll \Lst$, the larger nonlocality scale $\Lst$ turns out to be proportional to the usual quantum-field-theory cut-off $\Lambda_\textsc{uv}$. Here we set $\Lambda_\textsc{uv}=\Lst\de_0$ for the $R$-$R$ term and $\Lambda_\textsc{uv}=\Lst$ for the $R_{\mu\nu}$-$R_{\mu\nu}$ term, where $\de_0$ is a dimensionless constant that should be fixed comparing the renormalization-group-invariant scales of Stelle's theory with the nonlocal scale $\Lst$. We will not need to perform this explicit calculation to get our main results. Note that the logarithmic nonanalyticity of the quantum corrections is exclusively due to the Landau singularities of the one-loop amplitude, which are the same of the local theory.

Finally, the coefficients $\b_R$ and $\b_{\rm Ric}$ are nothing but the finite numerical constants coming in front of the quantum corrections. \emph{They are not beta functions}, since the quantum nonlocal theory is divergence-free and there are no counterterms.

For the theory \Eq{action}, the term \Eq{LQ} is correct only in the intermediate regime \Eq{interm}, while at higher energies $E \gtrsim \Lst$ quantum corrections differ substantially from Stelle's theory. Some scalar-field examples are given in \cite{Smailagic:2004yy,Koshelev:2021orf,Jax}. In particular, at energy scales larger than any scale present in the theory, preliminary results show that it should be possible to choose $V(E_{\mu\nu})$ in order to have amplitudes that fall off safely to zero in the UV after resummation \cite{Jax}. The other option is a strongly interacting theory in the UV but, since in the deep UV regime we are actually in the Weyl phase, high-energy physics will be affected only marginally by the details of the potential $V(E_{\mu\nu})$. 

As mentioned in section~\ref{Trans-Planckian-conformal-phase}, the logarithmic quantum corrections contribute to the two-point correlation function of the graviton when the Weyl symmetry is spontaneously broken, i.e., for 
$\varphi\lesssim \Mpl$. Indeed, expanding the quantum log corrections for small $\varphi/\Mpl$ we get
\ba
\ln \left(- {\rm f}^{-2} \B_{\phi^2 \bm{\hat g}} \right) &=& \ln \left[-{\rm f}^{-2}\B_{(\Mpl/\sqrt{2} + \varphi)^2 \bm{\hat g}} \right]\nn
&=& \ln \left(\frac{-2\B_{\bm{\hat g}}}{{\rm f}^2 \Mpl^2} \right) \left[1 + O\left(\frac{\varphi}{{\rm f}\Mpl}\right)\right] \nn
&=& \ln \left(\frac{-\B_{\bm{\hat g}}}{\Lst^2} \right) + O\left(\frac{\varphi}{\Lst}\right), \label{logSSB}
\ea
which coincides with the log operator in \Eqq{LQ} up to terms $O(\varphi/\Lst)$ that do not enter the two-point function of the graviton. In \Eqq{logSSB}, the first contribution, independent of the Goldstone boson, gives rise to the quasi-scale invariant CMB scalar spectrum \cite{Calcagni:2022tuz}, thus providing an authentic quantum origin of primordial perturbations.

Even if we are below the Planck energy in this regime, we are still fairly close to it. Perturbation theory remains valid as long as quantum corrections remain subdominant with respect to tree-level contributions, as imposed in \cite{Calcagni:2022tuz}. On top of that, the theory is asymptotically free \cite{Avramidi:1985ki,Modesto:2014lga,Briscese:2019twl}, which is a further guarantee that at ultrashort scales the perturbative expansion is well defined.


\subsection{Initial conditions in the trans- and post-Planckian phases}\label{incon}

Having described the main characteristics of the Weyl and the Higgs phase, let us discuss their spacetime metrics.

In the trans-Planckian regime, all conformally equivalent background solutions are physically equivalent to one another, and timelike and spacelike distances do not have any physical meaning. There is no physical singularity in the past, as discussed in section \ref{Singularity Problem}. 

After Weyl symmetry is broken, one of the conformally equivalent backgrounds becomes physical. The cosmological picture stemming from a theory claimed to be fundamental should be able to fully specify and justify theoretically the initial conditions of the Higgs phase, which in turn depend on what class of conformally equivalent metrics is picked during the trans-Planckian phase. Here we give a partial answer to this theoretical problem. The logical steps (A)-(B)-(C) will be the following:

\bigskip
\hspace{-.9cm}\begin{tabular}{cclclcl}
\multirow{3}{1.5cm}{Weyl invariance} & \multirow{3}{*}{$\Longrightarrow$\!\!\!\!} & (A) homogeneity $\longrightarrow$ Bianchi I,V,IX\!\!\!\! & \multirow{2}{.8cm}{\!\!$\left.\rule{0cm}{.7cm}\right\}\!\!\!\!\longrightarrow$} & \!\!\!\!Minkowski & \multirow{3}{.1cm}{\!\!\!\!\!\!\!$\left.\rule{0cm}{.9cm}\right\}\!\!\!\!\longrightarrow$}& \multirow{3}{.8cm}{$\begin{matrix}\textrm{Minkowski} \\ \textrm{FLRW}_{\textsc{k}=0}\end{matrix}$}\\
& & (B) isotropy & & \!\!\!\!$\textrm{FLRW}_{\textsc{k}=0,\pm1}$ & &\\
& & (C) $\textrm{FLRW}_{\textsc{k}=\pm1}\equiv\textrm{FLRW}_{\textsc{k}=0}$\!\!\!\! & & & &
\end{tabular}
\bigskip

\noindent {\bf (A)} In the trans-Planckian phase, Weyl invariance and the resulting absence of physical meaning to distances imply that no horizon can be defined to circumscribe regions in causal contact internally. Since there is neither an horizon nor a beginning at some finite time in the past, all regions of the universe have already been in causal contact at the onset of the Higgs phase and had infinite time to reach thermal equilibrium. This implies \emph{homogeneity} of the matter sector and one enters the Higgs phase with a metric compatible with this matter content. Asymptotically, this metric must be one of the Bianchi backgrounds \cite{LaL2,Hsu:1986ri,Krasinski:2003zzb,Kun03}, in particular, Bianchi I, V, or IX (homogeneous metrics with, respectively, zero, negative, and positive constant spatial curvature). Minkowski (static Bianchi I), FLRW (isotropic Bianchi I, V, or IX), and Kasner (anisotropic Bianchi I) are special cases of these metrics. 

\medskip

\noindent {\bf (B)} Weyl invariance during the trans-Planckian phase also reduces the number of independent scale factors to two and it further constrains the post-Planckian metric to be \emph{isotropic}. In fact, a Weyl transformation makes large and small distances the same and leaves the angles invariant. Hence, the geodesic distance between particles can be made arbitrarily small, so that everything was in causal contact in the Weyl phase. In particular, all the average values of the observables (including the components of the energy-momentum tensor $T_{\mu\nu}$) were the same at any point of the spatial section. Therefore, regions in the spatial section at different angles had the same temperature because particles at different angles were in contact in the past.

\medskip

\noindent {\bf (C)} To make the final step, we recall two well-known mini theorems. 
\begin{theo}
	Any FLRW line element with intrinsic curvature $\textsc{k}$,
	\ba
	\rmd \hat s^2 &=& \hat g_{\mu\nu} \rmd x^\mu\rmd x^\nu\label{RW}\\
	&=& - \rmd t^2 + a^2(t) \left[\frac{\rmd r^2}{1 - \textsc{k}\,r^2} + r^2 \left(\rmd \theta^2 + \sin^2 \theta \, \rmd \varphi^2 \right)\right], \nn
	\textsc{k} &=& 0, +1, -1\non
	\ea
	(where $r$, $\theta$, and $\varphi$ are the spherical coordinates on the spatial section in four dimensions) is conformally flat. 
\end{theo}
\begin{proof}
	For all values of $\textsc{k}$, there exists a coordinate transformation $(t, \, r, \, \theta, \, \varphi)\to(T, \, \rho, \, \theta, \, \varphi)$ such that the line element \Eq{RW} turns into
	\ba
	&& \rmd\hat s^2 = \om^2(x) \left[- \rmd T^2 + \rmd \rho^2 + \rho^2  \left(\rmd \theta^2 + \sin^2 \theta \, \rmd \varphi^2 \right) \right]  , \nn
	&& \om(x) =\left\{\begin{matrix} \hspace{-.25cm}\om(T) &\quad\hspace{-.3cm} {\rm for} \quad \textsc{k} = 0 \\
		\om(T, \rho) &\quad {\rm for} \quad \textsc{k} = \pm 1\end{matrix}\right.,
	\label{RWafterC}
	\ea
	where:
	$$T=\int\frac{\rmd t}{\om}\,, \qquad \rho=\exp\int \frac{\rmd r}{r\sqrt{1 - \textsc{k}\,r^2}}\,.$$ 
	Therefore, the spacetime is conformally equivalent to Minkowski but, in general, it is not homogeneous because $\om(x)=a(t)\, r/\rho(r)$ may depend on the radial coordinate as well as on the time coordinate.
\end{proof}
For instance, the coordinate transformation for the case $\textsc{k}=+1$ can be found in \cite{Narlikar:1986kr,Iihoshi:2007uz}. 
The case $\textsc{k} = -1$ can be worked out similarly \cite{Iihoshi:2007uz,Ibison:2007dv}. 

Another way to prove the theorem is to notice that the Weyl tensor vanishes exactly for the metric \Eq{RW}, $C_{\mu\nu\rho\s}=0$.

In the previous theorem, we did not call upon Weyl invariance and the above result applies to any diffeomorphism (Diff) invariant theory. In general, conformal equivalence of two metrics does not mean physical equivalence. Let us consider now a Weyl-invariant gravitational theory, i.e., a theory invariant under the symmetry Weyl$\times$Diff.
\begin{theo}
	In a Weyl$\times$Diff invariant theory, any FLRW solution \Eq{RW} with $\textsc{k}=\pm1$ is physically equivalent to the FLRW spacetime with $\textsc{k} = 0$. 
\end{theo}
\begin{proof}
	According to Theorem 1, we can first make a coordinate transformation to map \Eqq{RW} into the inhomogeneous, conformally flat line element \Eq{RWafterC} and, afterwards, we can make a Weyl transformation to a homogeneous and spatially flat FLRW spacetime. The explicit Weyl transformation reads
	\be\label{RWafterC2}
	\hat g^\prime_{\mu\nu} = \Omega^2 \hat g_{\mu\nu} \, , \qquad \Omega^2 = \frac{a^2(T)}{\om^2(T,\rho)}\,, 
	\ee
	where $\hat g_{\mu\nu}$ is given in \Eqq{RWafterC}. Therefore, in a Weyl-invariant theory any FLRW solution is equivalent to a spatially flat ($\textsc{k}=0$) FLRW solution.  
	
	On the other hand, the dilaton, which is a scalar under Diff, transforms under Weyl (according to \Eqq{WI}) into an inhomogeneous field violating the empirically validated Cosmological Principle of homogeneity and isotropy. However, the dilaton is not present in the geodesic equation \Eq{affine} for massless particles in affine parametrization, since Weyl covariance means
	\be
	D^2_{\la}[\bm{\hat g}]\,x^\s = 0\qquad\Longrightarrow\qquad D_{\la^*} [\bm{\hat g}^*]\,x^\s = 0\,,
	\ee
	which does not involve the dilaton. Hence, according to this Weyl covariant geodesic principle, an FLRW universe with $\textsc{k} = \pm1$ is gauge equivalent to an FLRW universe with $\textsc{k}=0$. 
\end{proof}
The details of both theorems are independent of the underlying theory except for the requirement (implicit in Theorem 2) of finiteness, without which one cannot invoke exact Weyl invariance at the classical as well as at the quantum level.

Having established that FLRW universes with $\textsc{k}=0,\pm 1$ all stay on the same gauge (Weyl$\times$Diff) orbit, it is immediate to see that $\textsc{k}=0$ is the preferred choice. To understand this, let us make an analogy with electromagnetism. Because of ${\rm U}(1)$ gauge invariance, the longitudinal polarization of the photon can be fixed to zero. Indeed, such polarization changes under gauge transformations, while the transverse polarizations are gauge invariant. Thus, the longitudinal polarization of the photon in electromagnetism is the analogue of the intrinsic curvature $\textsc{k}$ in the FLRW universe in Weyl gravity. To observe $\textsc{k}=\pm 1$ would be as unphysical as observing the longitudinal polarization of the photon.

The conclusion of the above arguments is that the metric at the onset of the Higgs phase is either Minkowski or flat FLRW, which are among the maximally symmetric ``ground states'' of conformal gravity \cite{Antoniadis:1985fd}. These arguments do not come full circle empirically because they do not fix the post-conformal Universe to be expanding and they still allow a contracting or a static metric. Still, they go a long way to set the post-Planckian initial conditions. Note that our conclusions, which are mainly based on quantum Weyl invariance, fully agree with \cite[section 12.7]{Narlikar:1986kr}, where the quantum cosmological transition probability of the universe was computed in Einstein gravity in the presence of conformal fluctuations. The setting there is different, but it is heavily affected by the properties of the Weyl rescaling as in our case. 


\section{Solutions to the problems of the hot big bang}\label{ProblemsCosmology}

The Weyl phase described in section \ref{Trans-Planckian-conformal-phase} can serve as an alternative to inflation to solve all the main problems of the standard hot-big-bang model of the early universe: the big-bang or singularity problem, the flatness problem, the horizon problem, and the monopole problem, plus the trans-Planckian problem. Since perturbations are generated after the Weyl phase, the reason why CMB temperature anisotropies exist and are as small as $\de T/T\sim 10^{-5}$ finds an answer in the Higgs phase described in section \ref{higsphs}; this is discussed in the companion paper \cite{Calcagni:2022tuz}. In this section, we will discuss the first stage of this two-step scenario alternative to inflation.


\subsection{Singularity problem and past completeness}\label{Singularity Problem}

The singularity issue is not whether a particular background is geodesically complete or not, but whether broad classes of spacetimes (for instance, with a cosmological expansion or with spherical symmetry) are geodesically complete, either in general or in a given theory of gravitation. In particular, one can discover examples of bouncing cosmologies in any otherwise pathological theory, and the question is how typical these solutions are. In the case of Einstein gravity, geodesic incompleteness is established through several sophisticated theorems, the simplest and most powerful being the Borde--Guth--Vilenkin (BGV) theorem \cite{BGV}. According to it, any universe that expanded in average during most of its history has past-incomplete worldlines. This is not enough to establish the existence of a global singularity such as the big bang, for which past-incompleteness should hold for all observers at the same time. However, it is strong evidence in favor of such a conclusion, since it is based on a background-independent proof which does not rely on any specific metric (such as FLRW), nor on any specific dynamical theory, nor on inflation.

In the case of Weyl gravity, the BGV theorem does not apply because, in the trans-Planckian phase, it is not even possible to talk about expanding backgrounds. The average expansion condition at the core of the BGV theorem cannot be met in a theory where the dynamics is Weyl-invariant. The expansion condition is not a conformally invariant statement and, vice versa, a Weyl transformation can always make an expanding background static or even contracting. To put it simply, the big-bang singularity issue becomes meaningless. 

This solution of the singularity problem is not based on any particular background in Weyl gravity but can be better appreciated when combined with the results on the Higgs phase described in section \ref{incon}. In the Weyl phase, the theory \Eq{action} in terms of the fields \Eq{MatterConfInv} is invariant under the Weyl transformation \Eq{WI}. Therefore, for any specific transformation \Eq{WI} with given rescaling $\Om^2=S(x)$, both the pair $(\hat{g}_{\mu\nu}, \phi)$ and the rescaled fields
\ba
\hat{g}^*_{\mu\nu} \coloneqq S(x)\, \hat{g}_{\mu\nu} \,, \qquad \phi^* \coloneqq S^{-\frac12}(x)\, \phi \,, 
\label{solStar}
\ea 
solve the equations of motion exactly. In particular, as evident from \Eqq{LEOM0}, the pair $(\hat{g}_{\mu\nu}, \phi)=(\hat{g}_{\mu\nu}^{\rm FLRW},\Mpl/\sqrt{2})$ is an exact solution of the equations of motion for the theory \Eq{action}, where $\hat{g}_{\mu\nu}^{\rm FLRW}$ is the flat FLRW metric in conformal time $\t=\int\rmd t/a(t)$ (where $t$ is proper time) defined by the line element
\be
\rmd \hat s^2_{\rm FLRW} = \hat{g}_{\mu\nu}^{\rm FLRW}\rmd x^\mu\rmd x^\nu = a^2(\t) \left(-\rmd\t^2 +\rmd\bm{x}^2\right),\label{FLRWeta}
\ee
where $a$ is the cosmological scale factor. If we take the rescaling $S(x) = S(\t) = 1/a^2(\t)$, the solution \Eq{solStar} reads
\be
\hat{g}^*_{\mu\nu} = 
 \eta_{\mu\nu}\,, \qquad \phi^* = \frac{\Mpl}{\sqrt{2}}\,a(\t)\,, 
\label{solMink}
\ee
which is Minkowski spacetime in the presence of a nontrivial dilaton profile. In other words, in a Weyl-invariant theory the cosmic expansion can be reinterpreted as a manifestation of Weyl invariance and, in particular, as the dynamics of the dilaton field, while leaving physical predictions unaltered.\footnote{This fact was already noted in classical scalar-tensor models concerning the equivalence of the Einstein and the Jordan frame \cite{Armendariz-Picon:2002jez,Catena:2006bd,Deruelle:2010ht}.}

However, as proved in appendix \ref{appeC}, the dilaton decouples from the geodesic equation of massless particles. Therefore, photons and any other massless particles move in Minkowski spacetime, which is geodesically complete (appendix \ref{GeoCompMink}). Indeed, calling $\la$ the affine parameter along the geodesic and $\xi_0$ a conserved scalar quantity made of one of the Killing vectors, the main geodesic equation
\be
\xi_0  = - \frac{\rmd\t}{\rmd\la}
\ee
is well-defined for $\t \in (-\infty, +\infty)$ and $\lambda \in (-\infty, +\infty)$, contrary to the analog geodesic equation \Eq{FLRWsing} in the FLRW background in Einstein gravity. 

This mechanism for solving the big-bang problem, and actually the wider singularity problem in general relativity, is independent of the details of the theory, but we stress the necessity to have a finite theory to get exact Weyl invariance in the deep UV. This fact was applied also in \cite{Bambi:2016wdn} to black holes.

Finally, it is interesting to stress a self-consistency cross-check of all the above results: the theory naturally realizes Penrose's Weyl curvature hypothesis \cite{Penrose:1979azm,Pen86,Pen11} (see \cite{Barrow:2002is,Stoica:2012gg,Hu:2021pfh,Kiefer:2021zqz} for reviews) and the beginning of the Universe is Weyl flat. At a singularity, in general, both the Kretschmann invariant $R_{\mu\nu\s\t}R^{\mu\nu\s\t}$ and the Weyl quadratic operator $C_{\mu\nu\s\t}C^{\mu\nu\s\t}=R_{\mu\nu\s\t}R^{\mu\nu\s\t}-2R_{\mu\nu}R^{\mu\nu}+R^2/3$ blow up \cite{Goode:1985ab}. However, the Weyl tensor $C_{\mu\nu\s\t}$ is conformally invariant and a Weyl transformation cannot get rid of singularities such that $\sqrt{|g|}\,C_{\mu\nu\s\t}C^{\mu\nu\s\t}=\infty$. The failure of the BGV theorem as a background-independent result in asymptotically local quantum gravity implies that the initial conditions set by Weyl symmetry are indeed special and that, in particular, the Weyl curvature must be finite (no singularity) or even vanishing. But this is exactly the conclusion we reach from the discussion in section \ref{incon}, which singles out Minkowski and flat FLRW as the metrics emerging asymptotically from the trans-Planckian phase. These spacetimes have $C_{\mu\nu\s\t}=0$. Perturbations of Minkowski or FLRW have, in general, nonvanishing Weyl tensor but, on one hand, those generated during the Weyl phase can be rendered regular by a Weyl transformation\footnote{If a spacetime with metric $\hat g_{\mu\nu}$ is geodesically incomplete, then a Weyl transformation $\hat g^*_{\mu\nu}=S(x)\,\hat g_{\mu\nu}$ makes it geodesically complete (\Eqq{solStar}) and the same procedure can be applied to the perturbed metric $\hat g^*_{\mu\nu}+h_{\mu\nu}$, so that the Weyl-equivalent metric $\hat g^{**}_{\mu\nu}=S'(x)(\hat g^*_{\mu\nu}+h_{\mu\nu})$ is made regular again, with finite Weyl curvature. This can be achieved by choosing the conformal factor $S'(x)$ such that to compensate the most singular among the components of $h_{\mu\nu}$ and to make spacetime geodesically complete.} and, on the other hand, perturbations generated in the post-Planckian phase \cite{Calcagni:2022tuz} are harmless because they evolve at times and distances away from the Weyl regime and from the ultrashort scales where singularities would usually form.


\subsection{Flatness problem}\label{FlatProblem}

The flatness problem is solved universally in Weyl gravity. According to the results of section \ref{incon}, if in the Weyl phase we take the FLRW metric as an exact or an asymptotic solution (more on this in section \ref{sotns}), then it must have $\textsc{k}=0$, while analyticity of the metric, namely of $a(t)$ and of the spatial metric $g_{ij}$, forces $\textsc{k}$ to be zero also in the Higgs phase when Weyl invariance is broken. To put it simply, we only have to extend the solution analytically to all values of proper time $t$ in the future of the Weyl phase.\footnote{We can regard the spontaneous breaking of the continuous Weyl symmetry as associated to a phase transition (e.g., \cite[Chap.~8]{PeSc}). 
 The facts that the solutions of the equations of motion do not have any discontinuity at the time $\tpl$ of symmetry breaking and that, as we will see later, the scaling of two-point functions also does not change abruptly suggest that we are in the presence of a second-order phase transition. In this case, the change in the metric cannot be abrupt either, and one cannot suddenly pass from a class of flat backgrounds to a nonflat one.} Notice that, although in the non-Weyl phase we can make a Diff transformation to change $\textsc{k}=0$ to $\textsc{k}=\pm 1$, we do not have at our disposal the Weyl symmetry to end up again with an FLRW spacetime.

Weyl invariance resolves the fine tuning of the standard hot-big-bang model not only at Planckian scales, but also, by analyticity of the metric, at lower energy scales. The observational constraint on the value of the curvature density parameter $\Om_\textsc{k}\propto \textsc{k}$ so close to zero is completely consistent and devoid of fine tuning in a Weyl-invariant theory because, if $\textsc{k}=0$ in the early evolution of the universe, then it must be $\textsc{k}=0$ at any other later time. 

\subsection{Horizon problem}

In our model for the early Universe, the horizon problem is solved simply because, in the Weyl phase, spacetime distances do not have any physical meaning: large and small distances are actually the same. Hence, as we noted in section \ref{incon}, in this phase there is no horizon separating regions in causal contact from those without.

Such a causal structure can be explicitly shown in asymptotically flat or bouncing backgrounds such as those considered in section~\ref{sotns}, which are all physically equivalent in the Weyl phase. In the examples \Eq{BDM} or \Eq{Bounce}, the scale factor $a(t)$ is almost always smaller than the Hubble radius $H^{-1}$. For the solution \Eq{BDM}, $a\leq H^{-1}$, while for the solution \Eq{Bounce} it is always possible to find an instant $\bar t$ before which the scale factor $a(t)$ is smaller than the Hubble horizon. Hence, for $t$ small enough, the distance between particles in all the solutions presented below is always smaller than the Hubble radius and the particles in the Universe can interact with one another without violating causality. 

As shown in section \ref{incon}, the key to the choice of initial conditions after the Weyl phase lies in the solution to the horizon problem, which is also called homogeneity problem because it implies that homogeneity in the early universe (which gives rise to large-scale homogeneity and isotropy at late times) should be encoded in the initial conditions \cite{LiL}. In the case of inflation, the horizon problem and the problem of initial conditions are treated separately and one has to make two different assumptions: a) the slow-roll mechanism, which solves the horizon problem, and b) chaotic/eternal inflation, where one can obtain nearly homogeneous and isotropic inflationary patches from an arbitrarily inhomogeneous and anisotropic metric \cite{Ste83,Vilenkin:1983xq,Linde:1983gd,Linde:1986fd,Linde:1986fc,Starobinsky:1986fx,Win09}. In the case of asymptotically local quantum gravity, the solution to the horizon problem also solves the problem of initial conditions, at least in part. 


\subsection{Monopole problem}
 
Since we have at our disposal a well-defined theory at any energy scale, we do not need to extend the Standard Model of particle physics to a grand unified theory (GUT). The theory \Eq{action} is complete regardless of the form of the local action \Eq{localS} (provided $\cL_{\rm loc}$ does not have instabilities), which we can assume to be the Standard Model. In other words, the theory \Eq{action} is already a unified theory of all fundamental interactions. Here, the concept of unification is understood as a finiteness regime in which all interactions vanish. Indeed, according to a preliminary study \cite{Jax}, in the finite theory \Eq{action} all the scattering amplitudes $\cA$ are suppressed as the inverse of a polynomial in the UV regime,  $\cA \sim 1/k^{2n}$ ($n \in \mathbb{N}$). Therefore, there is no unification strictly speaking, but something closer to the scenario presented in \cite{Robinson:2005fj}, in which all the couplings run to zero in the UV regime. 

However, if we insist that the action \Eq{action} must describe a UV-finite and Weyl-invariant completion of a GUT, we must deal with the monopole problem. By construction, monopoles are exact solutions of the theory \Eq{action} whenever they are so for the local theory \Eq{localS}. In a GUT, monopoles originate from the spontaneous symmetry breaking of the gauge group at some scale $M_X$ because the expectation value of the Higgs field $\Phi$ is, in general, different in causally disconnected domains and can take the value $\langle \Phi \rangle =0$ in some regions and $\langle \Phi \rangle = v\neq 0$ in other regions. However, in asymptotically local quantum gravity the GUT scale lies in the trans-Planckian regime where the theory is Weyl invariant and there are no causally disconnected domains. Therefore, $\langle \Phi \rangle$ will be the same in the whole Universe in the Weyl phase.


\subsection{Trans-Planckian problem}\label{Trans-Planckian-problem}

Finally, we make a comment on the cosmological trans-Planckian problem \cite{Martin:2000xs,Brandenberger:2000wr} in asymptotically local quantum gravity. 

In standard cosmology, in order to get nearly-Gaussian scale-invariant primordial perturbations, it is essential to set the initial state for the perturbations to be the Bunch--Davis vacuum. The latter coincides with the Minkowski solution in the infinite past, when the perturbative modes are deep inside the horizon. Therefore, sufficiently back in the past, the equations for perturbations reduce to the ones in Minkowski spacetime. However, in our theory the cosmological metric is only conformally equivalent to the Minkowski metric, hence, the wavelength $\lambda\sim a(t)$ of the perturbations will be trans-Planckian at some instant $t$ in the past, namely, $\lambda \lesssim \lp$ for $t \lesssim \tpl$. From the effective field theory point of view, it is hard to explain why higher-derivative operators, at the linear and the nonlinear level in the perturbations, do not affect the observed spectrum when $\lambda \simeq O(\lp)$. Indeed, operators like
\be
E_i(\Phi_i) \frac{\B^n}{\Mpl^{2n}} E_i(\Phi_i)\,,
\label{Higher order term}
\ee
will be of the same order of magnitude as the Einstein--Hilbert term $R$ at the Planck energy scale. 

This puzzle has a simple solution in nonlocal gravity \Eq{action}, whose classical equations of motion have the structure \Eq{LEOM0}.
 The terms $O(E_i^2)$ are subleading because all the scattering amplitudes vanish at the quantum level in the UV regime, while the solutions of the linear part of the equations of motion \Eq{LEOM0}, 
\be
\rme^{\bar\H(\Hes)^{(0)}} E_i^{(1)} = 0 \, , 
\label{linear0}
\ee
where $^{(0)}$ and $^{(1)}$ denote, respectively, the background and the first order in the perturbation, are the same of the local Einstein theory coupled to matter. Nonlinear operators in the perturbations are suppressed because the $n$-point amplitudes vanish in the UV regime, as we saw in section~\ref{Trans-Planckian-conformal-phase}. Therefore, no classical higher-derivative operators such as \Eq{Higher order term}, which are quadratic in the curvature in the action, contribute to the linearized classical equations of motion \Eq{linear0}.


\section{Background solutions} \label{sotns}

In this section, we describe possible choices of the background metric at the onset of the Higgs phase; according to the results of section \ref{incon}, in the asymptotic future in the Weyl phase the metric can be either Minkowski or $\textsc{k}=0$ FLRW.

A conformal rescaling of the Minkowski background $\eta_{\mu\nu}$ generates infinitely many possibilities among which we select two possible classes of scenarios: (I) asymptotically flat spacetimes and (II) bouncing spacetimes. For scenario (I), we identify three subcases that we label (Ia), (Ib), and (Ic). Scenarios (Ia) and (Ib) are asymptotically Minkowski in the past and in the future, and FLRW at intermediate times. Scenario (Ic) is asymptotically Minkowski in the past and asymptotically FLRW in the future. Scenarios of type (II) are FLRW. A characteristic common to all these solutions is that they all describe a causally connected spacetime where $Ha\leq 1$ at all times in $\tpl=1$ units.


\subsection{Asymptotically flat backgrounds}

Taking again the analogy mentioned in section~\ref{incon} with the gauge fixing of the longitudinal polarization mode of the photon, one can regard Minkowski spacetime as a unique ``zero choice'' selected by conformal symmetry during and at the end of the trans-Planckian phase. This is the condition in the asymptotic past we will explore in this subsection. In the asymptotic future in the post-Planckian Higgs phase, one can choose either Minkowski or $\textsc{k}=0$ FLRW. Then, by continuity, we can take any metric interpolating between the two asymptotically. 

In \emph{scenario (Ia)}, the metric in conformal time is obtained making a rescaling \Eq{WI} with the choice \cite[section~3.2]{BiDa} 
\be
\Om^2 = a^2(\t) = A + B \tanh \frac{\t}{\t_\Pl}\,,
\label{Ceta}
\ee
where $A$ and $B$ are constants, so that the line element with the metric \Eq{solMink} turns into
\ba
\rmd s^{* 2} &=& - \left( A + B \tanh \frac{\t}{\t_\Pl}\right) \left( - \rmd\t^2 + \rmd\bm{x}^2 \right), \nn
\lim_{\t \rightarrow \pm \infty} a^2 &=& A \pm B\,, \qquad A \geq B\,,\label{BD}
\ea
where we assumed $A \geq B$ in order to have $a^2(\t)>0$ for all times. Therefore, the metric in \Eqq{BD} interpolates between two asymptotic Minkowski spacetimes for $\t \rightarrow  \pm \infty$.\footnote{Another metric interpolating between two Minkowski spacetimes in the infinite past and in the infinite future can be defined in terms of the error function, $a^2(\t)=A+B\,\erf(\t/\t_\Pl)$.} In particular, our Universe was not born from a big bang but from a Weyl phase in which geometry was described by Minkowski spacetime.

Another metric similar to \Eqq{BD} that characterizes \emph{scenario (Ib)} can be directly defined in proper time, 
\ba
\rmd s^{* 2} &=&   - \rmd t^2 + a^2(t)\, \rmd\bm{x}^2\,,\nn
a &=&  A + B \tanh \frac{t}{\tpl}, \label{BDM}\\
\lim_{t \rightarrow \pm \infty} a &=& A \pm B \, , \qquad A \geq B \,.\non
\ea
In order to show that in a Universe described by the interpolating metric \Eq{BDM} we avoid the horizon problem, we compute the Hubble radius $R_H$ and the particle radius $R_{\rm p}$, or their comoving counterparts $r_H$ and $r_{\rm p}$:
\ba
\hspace{-.5cm}R_H\!\!&\coloneqq&\!\!  a\,r_H\coloneqq\frac{1}{H} = \frac{a}{\dot{a}}\nn
\hspace{-.5cm}\!\!&=&\!\! \tpl \cosh^2\left(\frac{t}{\tpl}\right)\!\left(\frac{A}{B}+\tanh\frac{t}{\tpl}\right), \label{HubbleS1b}\\
\hspace{-.5cm}R_{\rm p}\!\! &\coloneqq&\!\! a\,r_{\rm p}\nn
\hspace{-.5cm}\!\!&\coloneqq&\!\! a(t)\int^t_{t_{\rm i}}\frac{\rmd t'}{a(t')}=(\t-\t_{\rm i})a\,,\label{errepi}\\
\hspace{-.5cm}\t\!\! & = &\!\! \int \! \frac{\rmd t}{a(t)} \nn
\hspace{-.5cm}\!\!& = &\!\!   \frac{A t - B \tpl  \! \ln \left[ A \cosh \left( \frac{t}{\tpl } \right) +B \sinh \left( \frac{t}{\tpl } \right) \right] }{A^2-B^2},\label{etat-1}
\ea
where the latter expression is valid for $A\neq B$. The particle radius is the actual definition of a causal region comprising all points in contact with the observer through light signals since an initial time $t=t_{\rm i}$. However, depending on the background it may be ill defined, as is the case here for $t<0$, and the Hubble radius provides an alternative criterion for causal contact.

The functions $a(t)$ and $R_H(t)$ are tangent at $t=0$, while $a < R_H$ for all $t \neq 0$ in $\tpl=1$ units. In particular, for $A>B$ the Hubble radius goes to infinity for $t\rightarrow \pm \infty$ (figure~\ref{fig2}), while for $A=B$ one has a degenerate spacetime in the asymptotic past (figure~\ref{fig3}),
\be
\lim_{t \rightarrow + \infty} R_H = + \infty\,, \qquad \lim_{t \rightarrow - \infty} R_H = \frac{\tpl}{2}\,.
\label{asymptote}
\ee
At the instant $t$, the wavelength of a perturbation is
\be
\lambda = a(t) \,\la_{\rm com}\,,\qquad \la_{\rm com}=\frac{2\pi}{|\bm{k}|}\,, 
\label{lambdaat}
\ee
where $\la_{\rm com}$ is the comoving wavelength and $|\bm{k}|$ is the comoving wavenumber. Therefore, the scale factor can be regarded as the proper wavelength with $\la_{\rm com}=1$: any perturbation with wavelength within the horizon would stay in the causally connected patch. Therefore, using the relative size with respect to the Hubble radius as a criterion, the distance between causally connected points is always smaller than the Hubble radius. Taking instead the particle radius with an initial time in the infinite past, one reaches the same conclusion, since in that case $R_{\rm p}=+\infty$. Therefore, all the matter in the Universe is located inside a casually connected domain. 
\begin{figure}
\bc
\includegraphics[width=8.5cm]{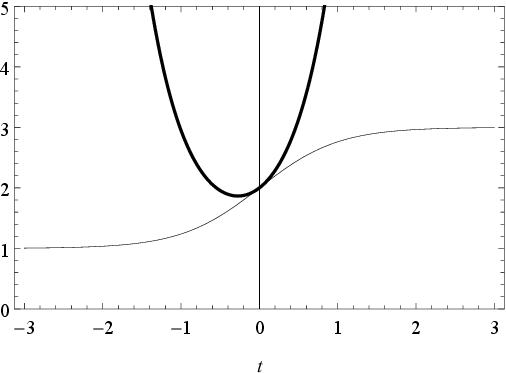}\\
\includegraphics[width=8.5cm]{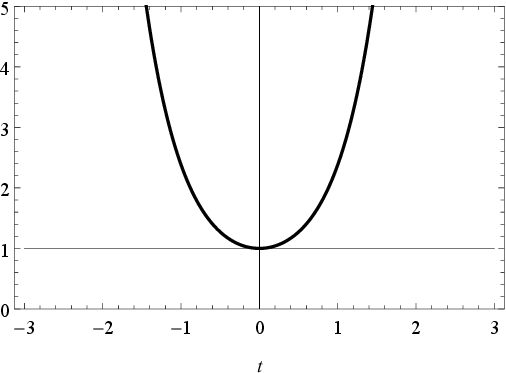}
\ec
\caption{Top: scale factor (proper wavelength) $a(t)$ \Eq{BDM} (solid thin curve) and Hubble radius $R_H$ \Eq{HubbleS1b} (solid thick curve) for $A > B$ and $\tpl=1$. Bottom: same as in the top panel for the comoving wavelength $\la_{\rm com}=1$ and the comoving Hubble radius $r_H$.}
\label{fig2}
\end{figure}
\begin{figure}
\bc
\includegraphics[width=8.5cm]{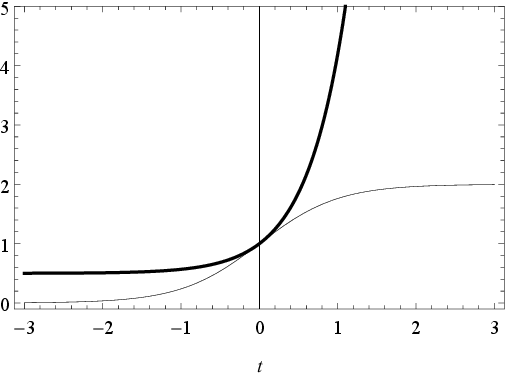}
\ec
\caption{Scale factor (proper wavelength) $a(t)$ \Eq{BDM} (solid thin curve) and Hubble radius $R_H$ \Eq{HubbleS1b} (solid thick curve) for $A=B$ and $\tpl=1$. The plot for comoving distances is similar to the one in the bottom panel of figure~\ref{fig2}.}
\label{fig3}
\end{figure}

Finally, in \emph{scenario (Ic)} the scale factor is
\bs\ba
a &=& \sqrt{A + \frac{t}{t_0} \, \rme^{\G (0,t/\tpl)}} \,, \label{KuzminConstant}\\
\lim_{t \rightarrow + \infty} a &=& \sqrt{\frac{t}{t_0}}\,,\label{KuzminConstant2}\\
\lim_{t \rightarrow - \infty} a &=& \sqrt{A}\,,
\ea\es
for which the Hubble radius is
\be
R_H = 2\frac{A t_0\, \rme^{-\G (0,t/\tpl)}+t}{1- \rme^{-t/\tpl}} 
> a \qquad \forall t \, . 
\label{HubbleKuzminConstant}
\ee
When $A=0$, the metric interpolates between a degenerate spacetime in the past and the solution of Einstein's theory for radiation (figure~\ref{fig4}), while for $A>0$ the metric approaches Minkowski spacetime in the past (figure~\ref{fig5}). This is the type of background described at the beginning of this section.
\begin{figure}
\bc
\includegraphics[width=8.5cm]{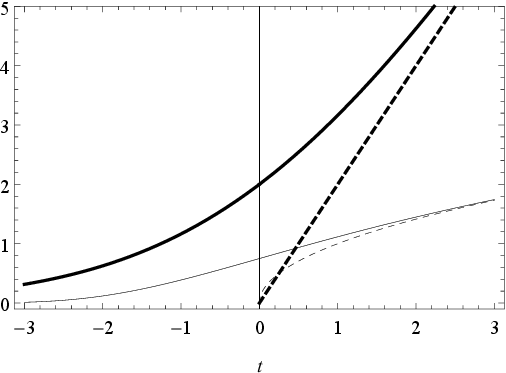}\\
\includegraphics[width=8.5cm]{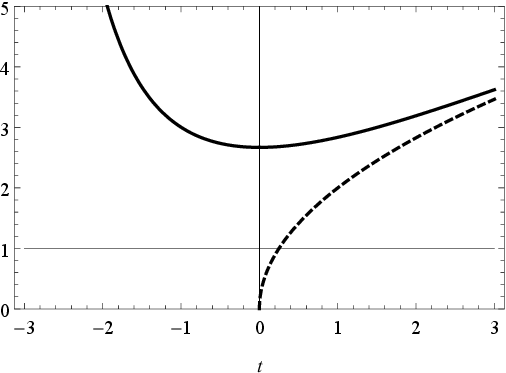}
\ec
\caption{Top: scale factor (proper wavelength) $a(t)$ \Eq{KuzminConstant} (solid thin curve) and Hubble radius $R_H$ \Eq{HubbleKuzminConstant} (solid thick curve) for $A=0$ and $t_0=\tpl=1$. The Einstein-gravity radiation-dominated solution \Eq{KuzminConstant2} and the corresponding $R_H=2t$ are shown as well (dashed thin and thick curves, respectively). Bottom: same as in the top panel for the comoving wavelength $\la_{\rm com}=1$ and the comoving Hubble radius $r_H$.}
\label{fig4}
\end{figure}
\begin{figure}
\bc
\includegraphics[width=8.5cm]{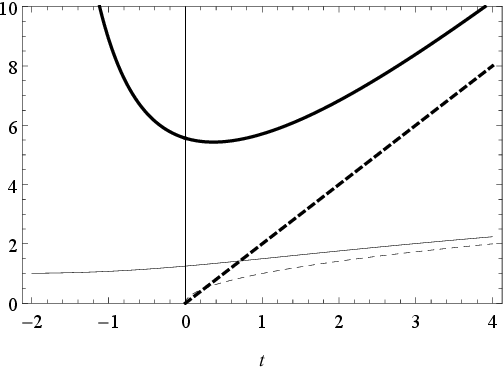}
\ec
\caption{Scale factor (proper wavelength) $a(t)$ \Eq{KuzminConstant} (solid thin curve) and Hubble radius $R_H$ \Eq{HubbleKuzminConstant} (solid thick curve) for $A=t_0=\tpl=1$. The general-relativity radiation-dominated solution \Eq{KuzminConstant2} with corresponding $R_H=2t$ are shown as well (dashed thin and thick curves, respectively). The plot for comoving distances is similar to the one in the bottom panel of figure~\ref{fig4}.}
\label{fig5}
\end{figure}


\subsection{Bouncing backgrounds}

\emph{Scenario (II)} is defined by a simple and quite general proposal for a bouncing Universe in proper time:
\bs\ba
a &=& \left[\left(\frac{t}{t_0}\right)^{2 n} + A^{2 n} \right]^{\frac{1}{4 n}},\qquad n \in \mathbb{N}^+,\label{Bounce}\\ 
\lim_{t \rightarrow \pm \infty} a &=& \sqrt{\frac{t}{t_0}} \,,\label{Bounce2} \\
\lim_{t \rightarrow 0} a &=& \sqrt{A} \, . 
\ea\es
Also in the bouncing Universe there is no horizon problem. For the metric \Eq{Bounce}, the Hubble radius is
\be
R_H = 2t\left[1+\left(\frac{At}{t_0}\right)^{2 n}\right]\,, 
\label{HubbleBounce}
\ee
which is divergent in $t=0$,
\be
\lim_{t \rightarrow 0^+} R_H = + \infty\, .
\ee
Therefore, $a(t)$ and $R_H(t)$ can cross each other never, once, or twice depending on the values of $t_0$ and $A$ (figures~\ref{fig6} and \ref{fig7}). In the first case, $a < R_H$ for all $t$; in the second case, $a<R_H$ for all $t$ before a certain point $\bar t$; in the third case, $a<R_H$ for $t>\tilde t$ or $t<\bar t$, where $0<\bar t<\tilde t$. We conclude that, in the bouncing Universe, there is always an instant $\bar t$ in the past before which particles had a chance to be in causal contact with one another.
\begin{figure}
\bc
\includegraphics[width=8.5cm]{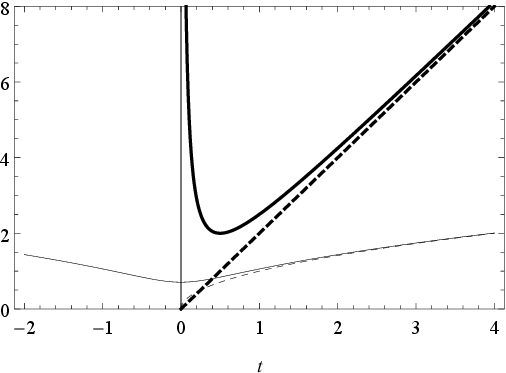}\\
\includegraphics[width=8.5cm]{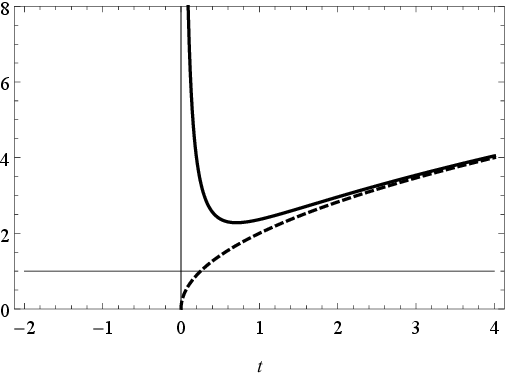}
\ec
\caption{Top: Scale factor (proper wavelength) $a(t)$ \Eq{Bounce} (solid thin curve) and Hubble radius $R_H$ \Eq{HubbleBounce} (solid thick curve) for $A=0.5$, $n=1$, and $t_0=1$. The general-relativity radiation-dominated solution \Eq{Bounce2} with corresponding $R_H=2t$ are shown as well (dashed thin and thick curves, respectively). In this plot, $a < R_H$ for all $t$. Bottom: Same as in the top panel but for the comoving wavelength $\la_{\rm com}=1$ and the comoving Hubble radius $r_H$.}
\label{fig6}
\end{figure}
\begin{figure}
\bc
\includegraphics[width=8.5cm]{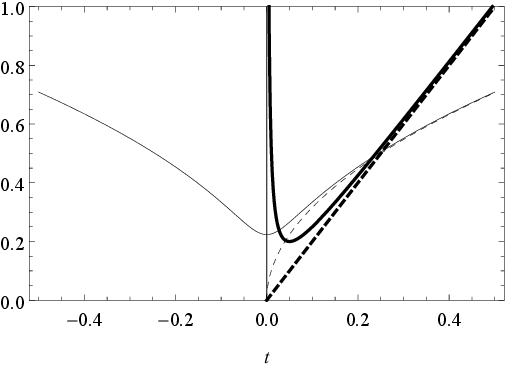}\\
\includegraphics[width=8.5cm]{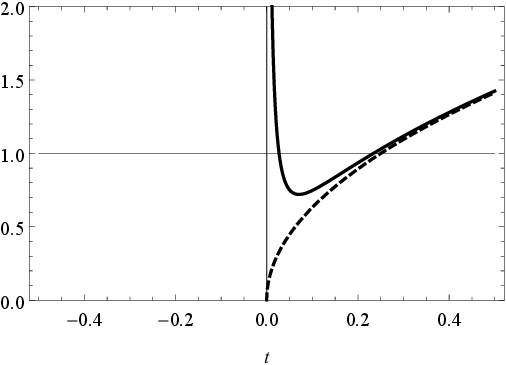}
\ec
\caption{Top: Scale factor (proper wavelength) $a(t)$ \Eq{Bounce} (solid thin curve) and Hubble radius $R_H$ \Eq{HubbleBounce} (solid thick curve) for $A=0.05$, $n=1$, and $t_0=1$. The Einstein-gravity radiation-dominated solution \Eq{Bounce2} with corresponding $R_H=2t$ are shown as well (dashed thin and thick curves, respectively). In this plot, we have an interval of time in which $a>R_H$. Bottom: Same as in the top panel but for the comoving wavelength $\la_{\rm com}=1$ and the comoving Hubble radius $r_H$.}
\label{fig7}
\end{figure}

From the scale factor \Eq{Bounce}, we can reconstruct the conformal time
\ba
\hspace{-1cm}\t &=& \int \frac{\rmd t}{a(t)}\nn
\hspace{-1cm}&=&\frac{t}{\sqrt{A}}{}_2F_1\left[\frac{1}{4 n},\frac{1}{2 n};1+\frac{1}{2 n};-\left(\frac{t}{At_0}\right)^{2 n}\right],\label{etat}
\ea 
where ${}_2F_1$ is the hypergeometric function.


\section{Discussion}\label{conclusions}

In this paper, we proposed a quantum-gravity scenario for the cosmology of the early Universe that passes through a Planckian Weyl-invariant phase. It is this Weyl symmetry that geometrically solves, in a simple and natural way, the cosmological problems of the standard hot-big-bang model and, at the same time, generates a scale-invariant power spectrum at large scales without the need of any inflationary era.

These findings rely on a recently proposed asymptotically local field theory for gravity coupled to matter consistent with unitarity, causality, and finiteness at the quantum level. Indeed, our statements above are based on the property of finiteness, that secures Weyl invariance at the classical as well as at the quantum level. Once Weyl symmetry is spontaneously broken, two-point correlation functions get perturbative quantum log-corrections that make primordial spectra deviate from exact scale invariance \cite{Calcagni:2022tuz}.

At the theoretical level on which we focussed in this work, the problem of the initial conditions of the Universe at the beginning of the Higgs phase is still partially open (section \ref{sotns}); within asymptotically local quantum gravity, it is not yet clear how an \emph{expanding} FLRW background emerges from the Weyl phase, although we have reduced the available choices to homogeneous and isotropic metrics with either expansion direction. A nonperturbative treatment of the removal of the big bang could help to better understand the role of quantum physics in the solution we propose, where classical Weyl invariance and its quantum preservation have been implemented in two somewhat separate steps. For example, obtaining the minisuperspace wavefunction of the Universe at early times in Hamiltonian formalism would encode both steps into one, albeit with the limitations inherent to that approach.
	
A much more pressing issue is to fill the details of the mechanism for the breaking of Weyl invariance. After sketching the differences a spontaneous Weyl symmetry breaking should or could have with respect to the Higgs mechanism, here we have given a general recipe to include a dilaton potential with absolute minima and provided two such examples. To claim to be complete in a structural sense, our approach should work out at least one of these examples. These and other open questions await to be addressed in the future. 


\section*{Acknowledgments}

The authors are supported by grant PID2020-118159GB-C41 funded by the Spanish Ministry of Science and Innovation MCIN/AEI/10.13039/501100011033. G.C.\ thanks A.\ De Felice and S.\ Kuroyanagi for useful discussions. L.M.\ is supported by the Basic Research Program of the Science, Technology, and Innovation Commission of Shenzhen Municipality (grant no.\ JCYJ20180302174206969).


\appendix


\section{Theory at the scale \texorpdfstring{$\Est$}{Lambdahd}} \label{appeA}

In this section, we derive the limit of the classical and the quantum theory in the Higgs phase, i.e., in the energy range \Eq{interm}. 


\subsection{Classical theory at the scale \texorpdfstring{$\Est$}{Lambdahd}} \label{appeA1}

The action in the UV limit reads
\ba
S_{\Est} \simeq \int \rmd^D x \, \sqrt{|g|}\, E_{\mu\nu} F(\Hes)_{\s\t}^{\mu\nu} E^{\s\t} \, , 
\ea
where the form factor $F(\Hes)$ is defined in \Eqq{FF} and $\exp\H(\Hes)$ approaches the polynomial \Eq{PolySimple3} in the UV regime. According to the results of section \ref{sec2a}, we can ignore the constant term in the form factor when $E\gtrsim \Est$, where $\Est=\sqrt{{\rm max}(1,b_0)/b}\,\Lst$. Then,
\ba 
2 \Hes F(\Hes) &=& \rme^{\bar{\H}(\Hes)} - 1\stackrel{\textsc{uv}}{\simeq} \rme^{\tilde{\g}_\textsc{e}} p(\Hes)\nn
& \stackrel{E\gtrsim \Est}{\simeq} & \rme^{\tilde{\g}_\textsc{e}} 
\left(-b\, z + z^4 \right) 
= \rme^{\tilde{\g}_\textsc{e}} 
\left[ -b\frac{4\Hes}{\Mpl^{2}\Lz^2} + \left(\frac{4\Hes}{\Mpl^{2}\Lst^2}\right)^4\right].\label{FUV}
\ea
Therefore,
\be
F(\Hes)= -\frac{2\rme^{\tilde{\g}_\textsc{e}}b}{\Mpl^{2}\Lz^2} + \frac{2^7\rme^{\tilde{\g}_\textsc{e}}}{(\Mpl^{2}\Lst^2)^4}\Hes^3\,.
\ee
Below $\Lst$, we can also ignore the $\Hes^3$ term, so that the gravitational sector of the classical action at the scale $\Est$ is
\ba 
S_{\Est}^{\rm vacuum}\! &=& -\frac{2\rme^{\tilde{\g}_\textsc{e}}b}{\Mpl^{2}\Lz^2} \int \rmd^D x \, \sqrt{|g|} E_{\mu\nu} E^{\mu\nu}\nn
\!&=&-\frac{2\rme^{\tilde{\g}_\textsc{e}}b}{\Mpl^{2}\Lz^2}\int \rmd^D x \, \sqrt{|g|} \, \frac{\Mpl^{2}}{2}\, G_{\mu\nu} \, \frac{\Mpl^{2}}{2} \, G^{\mu\nu} \nn
\!&=& -\frac{\rme^{\tilde{\g}_\textsc{e}}\Mpl^{2}b}{2\Lz^2}\int \rmd^D x \, \sqrt{|g|}  \left(R_{\mu\nu} \,  R^{\mu\nu}+ \frac{D-4}{4} R^2 \right). 
\label{StelleLimitD}
\ea
The full action also consists of the matter-matter and matter-gravity contributions coming from the operator $E_{\mu\nu} E^{\mu\nu}$ in \Eqq{action}. The total fourth-order action reads
\ba 
S_{\Est} &=& -\frac{2\rme^{\tilde{\g}_\textsc{e}}b}{\Mpl^{2}\Lz^2} \int \rmd^D x \, \sqrt{|g|} \, \left(\frac{\Mpl^{2}}{2}\, G_{\mu\nu} - \frac{1}{2}T_{\mu\nu}\right)\left(\frac{\Mpl^{2}}{2}\,G^{\mu\nu} - \frac{1}{2}T^{\mu\nu} \right)\nn
&=& -\frac{\rme^{\tilde{\g}_\textsc{e}}b}{\Mpl^{2}\Lz^2}\int \rmd^D x \, \sqrt{|g|}  
\left[\frac{\Mpl^{4}}{2}\left(R_{\mu\nu}\,R^{\mu\nu}+\frac{D-4}{4} R^2 \right) \right.\nn
&&\qquad\qquad\qquad\qquad\qquad\left.- \Mpl^{2} \left(R_{\mu\nu}-\frac{1}{2} g_{\mu\nu} R \right) T^{\mu\nu} 
+ \frac12 T_{\mu\nu} T^{\mu\nu}\right]\!. 
\label{StelleLimitDMatter}
\ea


\subsection{Quantum effective action at the scale \texorpdfstring{$\Est$}{Lambdahd}} \label{appeA2}

The gravitational action at the energy scale $\Est$ is actually Stelle gravity for a particular choice of the coefficients for the ${\bf Ric}^2$ and the $R^2$ operators in \Eqq{StelleLimitD}. According to \cite{Avramidi:1985ki,BOS}, the divergent part of the quantum effective action in Stelle gravity in $D=4$ dimensions reads
\be
\G^{\rm div} = - \frac{1}{2\ve} \frac{1}{(4 \pi)^2} \int \rmd^4 \sqrt{|g|} \left(
\b_2 C_{\mu\nu\rho\s} C^{\mu\nu\rho\s} + \b_3 R^2 \right)
\, , \qquad \frac{1}{\ve} = \ln \frac{\Lambda_\textsc{uv}^2}{\mu^2} \,, 
\label{Gdiv}
\ee
where $C_{\mu\nu\rho\s}$ is the Weyl tensor, $\b_{2,3}$ are beta functions, $\Lambda_\textsc{uv}$ is the cut-off energy scale in cut-off regularization, and $\mu$ is the renormalization energy scale.

Replacing $\ve$ with the cut-off expression in \Eqq{Gdiv}, the quantum effective action including the finite contribution at high energy reads
\ba
\G^{\rm div} + \G^{\rm finite} &=& - \frac{1}{2} \ln\frac{\Lambda_\textsc{uv}^2}{\mu^2} \frac{1}{(4 \pi)^2} \int \rmd^4 \sqrt{|g|} \left(\b_2 C_{\mu\nu\rho\s} C^{\mu\nu\rho\s} + \b_3 R^2 \right)\nn
&&+\frac{1}{2} \frac{1}{(4\pi)^2} \int \rmd^4x \sqrt{|g|} \left[\b_2 C_{\mu\nu\rho\s} \ln\left(-\frac{\B}{\mu^2}\right) C^{\mu\nu\rho\s}
   + \b_3 R \ln\left(-\frac{\B}{\mu^2}\right) R \right] \nn
&=&\frac{1}{2} \frac{1}{(4\pi)^2} \int \rmd^4x \sqrt{|g|} \left[\b_2 C_{\mu\nu\rho\s} \ln\left(-\frac{\B}{\Lambda_\textsc{uv}^2}\right)C^{\mu\nu\rho\s}
   + \b_3 R \ln\left(-\frac{\B}{\Lambda_\textsc{uv}^2}\right) R \right],\nn
\label{EGammaDivFin}
\ea
as one can calculate, for instance, via covariant perturbation theory \cite{Barvinsky:1987uw,Barvinsky:1990up,Barvinsky:1990uq,Barvinsky:1994cg} stopping at second order in curvatures \cite{Barvinsky:1990up}. Now we can use the identity \cite{Codello:2012kq}
\be
C_{\mu\nu\rho\s} C^{\mu\nu\rho\s} = 2 \left(R_{\mu\nu} R^{\mu \nu} - \frac{1}{3} R^2 \right) + O({\bf Riem}^3) 
\ee
to express the divergent and finite contributions to the quantum effective action \Eq{EGammaDivFin} in the Einstein basis up to higher curvature operators, i.e.,
\ba
\G^{\rm div} + \G^{\rm finite}   & = &
  \frac{1}{2} \frac{1}{(4 \pi)^2} \int \rmd^4x \sqrt{|g|} \left[ 2  \b_2 
   R_{\mu\nu} \ln\left(-\frac{\B}{\Lambda_\textsc{uv}^2}\right) R^{\mu\nu}
  - \frac{2}{3} \b_2 R \ln\left(-\frac{\B}{\Lambda_\textsc{uv}^2}\right) R\right.\nn
&& \qquad\qquad\qquad\qquad\quad\left. + \b_3 R \ln\left(-\frac{\B}{\Lambda_\textsc{uv}^2}\right) R \right]\nn
   &=& \frac{1}{(4 \pi)^2} \int \rmd^4x \sqrt{|g|} \left[  \b_2 
   R_{\mu\nu} \ln\left(-\frac{\B}{\Lambda_\textsc{uv}^2}\right) R^{\mu\nu}
  + \frac{3 \b_3 - 2 \b_2}{6} R \ln\left(-\frac{\B}{\Lambda_\textsc{uv}^2}\right)R\right].\nn\label{EGamma2}
\ea
These are well-known results for Stelle gravity but, after some modifications, \Eqq{EGamma2} is also the one-loop quantum effective action of the theory \Eq{action} in its local limit \Eq{StelleLimitD}. In contrast with Stelle's theory, the nonlocal theory considered in this paper is finite. This has two consequences. On one hand, we can identify the cut-off scale in \Eqq{Gdiv} with the fundamental scale of theory up to an $O(1)$ rescaling. On the other hand, the scale in $1/\ve$ may differ in the two operators in \Eqq{Gdiv}. In particular, and without loss of generality, we can set
\be
C^2\,{\rm term:}\quad\Lambda_\textsc{uv}= \Lst\,,\qquad\qquad R^2\,{\rm term:}\quad \Lambda_\textsc{uv} = \Lst\, \de_0\,,\qquad \de_0>0\,,
\ee
respectively for the Weyl-square and the Ricci-square operator. 

Finally, comparing \Eqqs{EGamma2} and \Eq{LQ}, we get the relation between the one-loop coefficients in the nonlocal theory \Eq{action} and the beta functions of Stelle gravity:
\be\label{betabeta}
\b_{\rm Ric} = \frac{\b_2}{(4\pi)^2} \, , \qquad \b_R = \frac{3 \b_3 - 2 \b_2}{6 (4 \pi)^2}\,. 
\ee
We reiterate that the left-hand sides are not beta functions coming from counterterms.


\section{Geodesic completeness for massless particles} \label{appeC}

For the sake of simplicity, in this section we change notation redefining $\hat{g}^*_{\mu\nu} \rightarrow \hat{g}_{\mu\nu}$ and $\phi^* \rightarrow \phi$. Therefore, for the pair $(\hat{g}_{\mu\nu}, \phi)=(\hat{g}_{\mu\nu}^{\rm FLRW},\Mpl/\sqrt{2})$ \Eqq{solStar} becomes
\ba\label{rescal}
\hat{g}_{\mu\nu} = S(x)\, g^{\rm FLRW}_{\mu\nu} \,, \quad \phi =  \frac{\Mpl}{\sqrt{2}}\, S^{-\frac12}(x)\,, 
\ea
where $g_{\mu\nu}^{\rm FLRW}$ is the FLRW metric given in \Eqq{FLRWeta}. 


\subsection{Geodesics for massless particles}

The action for photons or general massless particles, parametrized by a parameter $p$, is
\ba
S_{\g} = \int\rmd p\, \cL_\g = \int\rmd p\, e^{-1}(p)\phi^2\hat{g}_{\mu\nu}  \dot x^\mu \dot x^\nu\,,\qquad \dot{}\coloneqq \frac{\rmd}{\rmd p}\,,
\label{Lm0}
\ea
where $e(p)$ is an auxiliary field. The action \Eq{Lm0} is invariant under general coordinate transformations ${x'}^\mu=f^\mu(x^\nu)$, the Weyl rescaling \Eq{WI}, and reparametrizations $p^\prime= f(p)$ of the worldline, since $e(p)$ transforms as $\rme^\prime(p^\prime) = e(p) (\rmd p^\prime/\rmd p)^{-1}$.

The variation with respect to $e$ gives
\ba
\frac{\de S_{\g}}{\de e} = - \int \rmd p  \, \frac{\de e}{\rme^2} \phi^2 \, \hat{g}_{\mu\nu} \, \dot{x}^{\mu} \dot{x}^\nu = 0 \quad \Longrightarrow \quad   \rmd \hat{s}^2 =  \hat{g}_{\mu\nu} \rmd x^\mu \rmd x^\nu = 0\,,
\label{ds0}
\ea
consistently with the fact that massless particles travel along the light-cone.

Varying with respect to $x^\mu$, one gets the geodesic equation in the presence of the dilaton field. In the gauge $e(p)={\rm const}$,
\ba
D^2_p[\bm{g} = \phi^2 \bm{\hat g}]\, x^\s &\coloneqq& \frac{\rmd^2 x^\s}{\rmd p^2}+\G^\s_{\mu\nu}\,\dot x^\mu \dot x^\nu\nn
&=&\frac{D^2[\bm{\hat g}]\, x^\s}{\rmd p^2} + 2 \frac{\partial_\mu \phi}{\phi} \dot x^\mu \dot x^\s-\frac{\p^\s \phi}{\phi} \dot x_\mu \dot x^\mu = 0 \,,\label{geom0}
\ea
where
\be
\G^\s_{\mu\nu}\coloneqq \frac12 g^{\rho\s}\left(\p_{\mu} g_{\nu\rho}+\p_{\nu} g_{\mu\rho}-\p_\rho g_{\mu\nu}\right)\label{leci}
\ee
is the Levi-Civita connection for the metric $\bm{g}$. When we contract \Eqq{geom0} with the velocity $\dot x_\s$ and we use $\rmd\hat{s}^2 =0$ obtained in \Eqq{ds0}, we get the on-shell condition
\ba
\dot x_\s \, D^2_p[\bm{\hat g}]\, x^\s + 2  \dot x_\s\,\frac{\p_\mu \phi}{\phi} \dot x^\mu\dot x^\s 
-  \dot x_\s \, \frac{\p^\s \phi}{\phi} \dot x^\mu \dot x_\mu = 0 \quad \Longrightarrow \quad \dot x_\s\,D^2_p[\bm{\hat g}]\, x^\s = 0 \,.
\ea
Therefore, $D^2_p[\bm{\hat g}]\, x^\s$ must be proportional to the velocity,
\ba
D^2_p[\bm{\hat g}]\, x^\s = f \, \dot x^\s\,,  \qquad f = {\rm const}\,, 
\label{constre}
\ea
which is null on the light-cone. Under a reparametrization $q=q(p)$ of the worldline, \Eqq{constre} turns into
\ba
\frac{\rmd^2 x^\s}{\rmd q^2} + \hat\G^\s_{\mu\nu}  \frac{\rmd x^\mu }{\rmd q} \frac{\rmd x^\nu}{\rmd q} 
=  \frac{\rmd x^\s}{\rmd p} \left(\frac{\rmd p}{\rmd q}\right)^2 \left(f \frac{\rmd q}{\rmd p} - \frac{\rmd^2 q}{\rmd p^2} \right).
\label{parame}
\ea
Choosing the $p$-dependence of $q$ such as to make the right-hand side of \Eqq{parame} vanish, we end up with the geodesic equation in affine parametrization. Hence, we can redefine $q \rightarrow \lambda$ and, finally, get the affinely parametrized geodesic equation for photons in the metric $\hat{g}_{\mu\nu}$, 
\be
D^2_\la[\bm{\hat g}]\,x^\s = 0 \, .
\label{affine}
\ee
This result can be obtained also noting that the geodesic equation is invariant under a Weyl transformation, where $\la$ is the conformally transformed affine parameter \cite[Example 12.2]{Narlikar:1986kr}.

We can now investigate the conservations laws based on the symmetries of the metric. Let us consider the scalar
\ba
\xi\coloneqq \hat{g}_{\mu\nu} v^\mu \frac{\rmd x^\nu}{\rmd \lambda}=\hat{g}_{\mu\nu} v^\mu \dot x^\nu\,,
\label{alpha}
\ea
where $v^\mu$ is a generic vector and, from now on, $\dot{}=\rmd/\rmd\la$. Taking the derivative of \Eqq{alpha} with respect to $\lambda$ and using \Eqqs{leci} and \Eq{affine}, we get
\ba
\frac{\rmd \xi}{\rmd \lambda} = \frac{1}{2}  v^\mu \p_\mu \hat{g}_{\rho \nu} \dot x^\rho 
\dot x^\nu + \hat{g}_{\mu\nu} \p_\rho v^\mu \dot x^\nu \dot x^\rho 
= \frac{1}{2} [\mathcal{L}_v \hat{g}]_{\rho \nu}\, \dot x^\rho \dot x^\nu\,, 
\ea
where $\mathcal{L}_v \bm{\hat g}$ is the Lie derivative of $\hat{g}_{\mu\nu}$ by a vector field $v^\mu$. Thus, if $v^\mu$ is a Killing vector field, namely $\mathcal{L}_v \bm{\hat g}=0$, then the scalar \Eq{alpha} is affinely constant:
\ba
\frac{\rmd \xi}{\rmd \lambda} =\frac{\rmd}{\rmd \lambda} \left(\hat{g}_{\mu\nu} v^\mu \dot x^\nu\right)=0\,. 
\label{CL}
\ea


\subsection{Geodesic completeness in Weyl gravity}\label{GeoCompMink}

After the rescaling \Eq{rescal}, the exact solution is \Eqq{solMink}, Minkowski spacetime in the presence of the dilaton: 
\ba
\rmd\hat{s}^2 = - \rmd\t^2 + (\rmd x^1)^2 +(\rmd x^2)^2 + (\rmd x^3)^2 \,, \qquad \phi = \frac{\Mpl}{\sqrt{2}}\,a(\t)\,.
\ea
However, it is crucial to note that the dilaton does not appear in the geodesic equation \Eq{affine} and in the conservation equation \Eq{CL}, which is only a function of $\hat{g}_{\mu\nu}$. 

It is well known that Minkowski spacetime has ten Killing vectors, four of which are
\ba    
v_0^\mu=(1, \, 0, \, 0, \, 0)^\mu \, , \qquad 
v_1^\mu=(0, \, 1, \, 0, \, 0)^\mu \, , \qquad 
v_2^\mu=(0, \,  0, \,  1, \, 0)^\mu \, , \qquad 
v_3^\mu=(0, \,  0, \,  0, \, 1)^\mu \, . \nn
\label{killing}
\ea
When replaced in \Eqq{CL}, these vectors give the following conserved quantities:
\bs\ba
&& \xi_0 \coloneqq \hat{g}_{\mu\nu} \, v_0^\mu \dot{x}^\nu = \hat{g}_{00} \, \dot{\t} = - \dot{\t}
\, , \\
&& \xi_1 \coloneqq \hat{g}_{\mu\nu} \, v_1^\mu \dot{x}^\nu = \hat{g}_{11} \, \dot{x}^1 
=   \dot{x}^1
\, , \\
&& \xi_2 \coloneqq \hat{g}_{\mu\nu} \, v_2^\mu \dot{x}^\nu =  \hat{g}_{22}  \, \dot{x}^2 
 = \dot{x}^2 
\, ,  \\
&& \xi_3 \coloneqq \hat{g}_{\mu\nu} \, v_3^\mu \dot{x}^\nu =  \hat{g}_{33} \,  \dot{x}^3
= \dot{x}^3
 \, ,
\ea\label{e0123}\es
which are the geodesic equations for massless particles in Minkowski spacetime. Since \Eqqs{e0123} are well-defined everywhere for any value of the affine parameter $\lambda$, then this spacetime is null geodesically complete. The condition \Eq{ds0}, i.e., $\rmd\hat{s}^2=0$, only imposes the following consistency relation between the conserved quantities $\xi_0, \xi_1, \xi_2, \xi_3$ in \Eqq{e0123}:
\ba
\rmd \hat{s}^2 = 
 - \xi_0^2 + \xi_1^2 + \xi_2^2 + \xi_3^2 = 0\,,
\label{explids}
\ea
which simply states that massless particles travel on the light-cone. 


\subsection{Geodesic incompleteness in general relativity} \label{GeoIncFLRW}

In contrast with the case of conformal gravity, the FLRW spacetime in Einstein gravity only possesses three out of the four Killing vectors \Eq{e0123}:
\ba
\xi_1 = a^2(\t) \, \dot{x}^1 \,, \qquad \xi_2 = a^2(\t) \, \dot{x}^2  \,, \qquad \xi_3 = a^2(\t) \, \dot{x}^3 \,.
\label{e0123FLRW}
\ea
If we now replace the above conserved quantities in the FLRW line element before the rescaling, which is zero for light,
\ba
\rmd s^2_{\rm FLRW} = a^2(\t) \left[ -\rmd\t^2 + (\rmd x^1)^2 +(\rmd x^2)^2 + (\rmd x^3)^2 \right]=0, 
\ea
we get the geodesic equation for $a(\t)\neq 0$:
\ba
 \dot\t^2 =   \frac{\xi_1^2+\xi_2^2+\xi_3^2}{a^4(\t)}\,.
 \label{FLRWsing}
\ea
For the sake of simplicity, we assume the spatial section to be flat and that it is well-defined for $\t \neq 0$. We can restrict our attention to $\t > 0$ and integrate \Eqq{FLRWsing} for the case of a radiation-dominated universe, where $a(\t) = \t/\t_0$ and $\t_0$ is conformal time today. The result for $\dot\t>0$, $\t> 0$, and $\lambda >0$ is
\ba
\hspace{-1cm}\int \rmd\t\, \frac{\t^2}{\t_0^2} = \sqrt{\xi_1^2 + \xi_2^2 + \xi_3^2} \int \rmd \lambda &\quad \Longrightarrow \quad& 
\lambda = \frac{1}{3 \sqrt{\xi_1^2 + \xi_2^2 + \xi_3^2}}\frac{\t^3}{\t_0^2}\nn
&\quad \Longrightarrow \quad& 
\t = \left(3\t_0^2\sqrt{\xi_1^2 + \xi_2^2 + \xi_3^2} \right)^{\frac{1}{3}} \lambda^{\frac{1}{3}}\,. \label{etaPos}
\ea
On the other hand, for $\dot\t<0$, $\t>0$, and $\lambda<0$, we get
\ba
\hspace{-1cm}-\int \rmd\t\, \frac{\t^2}{\t_0^2} = \sqrt{\xi_1^2 + \xi_2^2 + \xi_3^2} \int \rmd \lambda &\quad \Longrightarrow \quad& 
\lambda = -\frac{1}{3 \sqrt{\xi_1^2 + \xi_2^2 + \xi_3^2}}\frac{\t^3}{\t_0^2}\nn
&\quad \Longrightarrow \quad &
\t = \left(3\t_0^2\sqrt{\xi_1^2 + \xi_2^2 + \xi_3^2} \right)^{\frac{1}{3}} (-\lambda)^{\frac{1}{3}}\,.
 \label{etaNeg}
\ea
Since in \Eqq{etaPos} the time coordinate $\t(\lambda)$ of the massless particle evolves from the big-bang singularity for a finite amount of the affine parameter $\lambda$, then spacetime is not geodesically complete. In other words, we do not have any information about the past evolution for $\lambda \leq 0$. The two solutions \Eq{etaPos} and \Eq{etaNeg} cannot be patched together in order to get a new solution defined for all $\lambda$, since the function $\t(\lambda)$ is nonanalytic in $\lambda = 0$. 


\end{document}